\documentclass[12pt,onecolumn]{IEEEtran}
\usepackage{times,color,amsfonts,amsmath,amsthm,amssymb,graphicx,epsf,psfrag,verbatim,epstopdf,epsfig}
\usepackage{setspace}
\def\boxit#1{\vbox{\hrule\hbox{\vrule\kern3pt
        \vbox{\kern3pt#1\kern3pt}\kern3pt\vrule}\hrule}}

\def\reals{ { {\rm  I \kern-0.15em R }  } }
\def\complex{ {\,{{\rm C} \kern-0.50em \raise0.20ex {  |}}\, }}

\def\Rbf{{\bf R}}

\def\Rxx{\Rbf_{\ssstyle X\kern-.1em X}}

\let\ssstyle=\scriptscriptstyle

\def\Kout{\setbox1=\hbox{\Huge\bf K}\hbox to
1.05\wd1{\hspace{.05\wd1}
}}
\def\Sout{\setbox1=\hbox{\Huge\bf S}\hbox to 1.05\wd1{\hspace{.05\wd1}
}}

\newtheorem{lemma}{Lemma}
\newtheorem{theorem}{Theorem}
\newtheorem{proposition}{Proposition}
\renewcommand\qed {$\blacksquare$}

\newtheorem{defn}{Definition}

\newtheorem{ex}{{\em Example}}

\newtheorem{remark}{Remark}

\newcommand{\ignore}[1]{}
\ifodd 0
\newcommand{\rev}[1]{{\color{blue}#1}}
\newcommand{\com}[1]{\textbf{\color{red}(#1)}} 
\newcommand{\clar}[1]{\textbf{\color{green}(NEED CLARIFICATION: #1)}}
\newcommand{\response}[1]{\textbf{\color{magenta}(RESPONSE: #1)}} 

\else
\newcommand{\rev}[1]{#1}
\newcommand{\com}[1]{}
\newcommand{\clar}[1]{}
\newcommand{\response}[1]{}
\fi

\begin{document}

\title{Sufficient Conditions on the Optimality of Myopic Sensing in Opportunistic Channel Access: \\ A Unifying Framework}   

\author{\IEEEauthorblockN{Yang Liu, Mingyan Liu and Sahand Haji Ali Ahmad\\
\thanks{Y. Liu and M. Liu are with the Electrical Engineering and Computer Science Department, University of Michigan, Ann Arbor, MI 48105, USA, {\{youngliu,mingyan\}@umich.edu}.  A preliminary version of this work appeared in Allerton 2009. The work was partially supported by the NSF under grants CIF-0910765 and CNS-1217689, and by the ARO under Grant W911NF-11-1-0532. } 
}} 

\maketitle

\begin{abstract}
This paper considers a widely studied stochastic control problem arising from opportunistic spectrum access (OSA) in a multi-channel system, with the goal of providing a unifying analytical framework whereby a number of prior results may be viewed as special cases.  Specifically, we consider a single wireless transceiver/user with access to $N$ channels, each modeled as an iid discrete-time two-state Markov chain.  In each time step the user is allowed to sense $k\leq N$ channels, and subsequently use up to $m\leq k$ channels out of those sensed to be available. 
Channel sensing is assumed to be perfect, and for each channel use in each time step the user gets a unit reward. 
The user's objective is to maximize its total discounted or average reward over a finite or infinite horizon.  
\rev{This problem has previously been studied in various special cases including $k=1$ and $m=k\leq N$, often cast as a restless bandit problem, with optimality results derived for a myopic policy that seeks to maximize the immediate one-step reward when the two-state Markov chain model is positively correlated. } 
In this paper we study the general problem with $1\leq m\leq k\leq N$, and derive sufficient conditions under which the myopic policy is optimal for the finite and infinite horizon reward criteria, respectively.  
It is shown that these results reduce to those derived in prior studies under the corresponding special cases, and thus may be viewed as a set of unifying optimality conditions. 
Numerical examples are also presented to highlight how and why an optimal policy may deviate from  the otherwise-optimal myopic sensing given additional exploration opportunities, i.e., when $m<k$. 
\end{abstract}

\vspace{1em}

\begin{IEEEkeywords}
Opportunistic Spectrum Access (OSA), POMDP, restless bandits, index policy, myopic policy, sufficient condition
\end{IEEEkeywords}

\newpage


\section{Introduction}\label{sec:intro}

We consider the following stochastic control problem: There are $N$ uncontrolled, independent and identically distributed, two-state discrete-time Markov chains, with the two states denoted ``$1$'' and ``$0$'' respectively, and the transition probabilities given by $p_{ij}$, $i, j = 0, 1$.  
The system evolves in discrete time.  In each time instance, a user selects exactly $k$ out of the $N$ processes and is allowed to observe their states.  The user is allowed to receive a unit reward from a process observed to be in state $1$, but the total reward is limited at $m, m\leq k$, at each step. The processes that the user does not select do not reveal their true states. The objective is to derive a selection strategy for the user so that its total expected discounted or average reward  over a finite or infinite horizon is maximized. 
This is a partially observed MDP (POMDP) problem \cite{puterman,Smallwood&Sondik:71OR} due to the fact that the states of the underlying Markov processes are not fully observed at all times; as a consequence the system state as perceived by the user is in the form of a probability distribution, commonly referred to as the {\em information state} or {\em belief state} of the system \cite{kumar}.  More specifically, this problem is an instance of the restless bandit problem with multiple plays \cite{whittle,anantharam87,agrawal90}.  

The above problem abstraction and a number of its variations have been quite extensively studied in the past few years in the context of multichannel opportunistic spectrum access (OSA), including \cite{Zhao&etal:08TWC,Ahmad09optimalityof,guha,guha2,Ahmad:2009:MOA:1793974.1794208}. 
Within this application, each Markov process represents a wireless channel in a discrete time setting, whose state transitions reflect dynamic changes in channel conditions caused by fading, interference, and so on, with state $1$ denoting a ``good'' or available state, and state $0$ the ``bad'' or unavailable state, in which communication may succeed and fail, respectively.  
A user wishing to transmit must first sense the state of a channel at the beginning of a time step, and can only transmit in that channel if it is sensed to be in the ``good'' state.  The user cannot sense more than $k$ channels, nor can it transmit in more than $m$ at a time.  Such constraints come from both hardware, e.g., the number of antennas available, and from performance requirements, e.g., channel sensing takes time so stringent delay requirement can limit the amount of sensing allowed. 
Finally, if all $k$ selected channels are in the ``bad'' state, the user has to wait till the beginning of the next time step to repeat the selection process. 
%
While this model captures some of the essential features of multichannel opportunistic access, it has the following limitations: the simplicity of the iid two-state channel model, and the implicit assumption that channel sensing is perfect and the lack of penalty if the user transmits in a bad channel due to imperfect sensing. 
Nevertheless, this model allows us to obtain analytical insights into the problem, and more importantly, insights into the more general problem of restless bandits with multiple plays. 

Prior work investigated various special cases of the model outlined above, henceforth referred to as the $(k,m)$ model.  
Specifically,  authors derived sufficient conditions for guaranteeing the optimality of a greedy/myopic sensing for the $(1,1)$ case, i.e., $k=m=1$ with $N=2$ in \cite{DBLP:conf/icc/JavidiKZL08}, with positively correlated channel model. \cite{Zhao&etal:08TWC} further proved the performance bounds of a greedy/myopic policy for this case (as well as negatively correlated channels) and \cite{Liu&Zhao:08Asilomar} proved the same for the $(N-1, N-1)$ case, while \cite{guha,guha2} looked for provably good approximation algorithms for a similar problem but relaxing the requirement that all Markov chains are identically distributed. The assumption of perfect sensing was relaxed in \cite{DBLP:conf/cdc/LiuZ09} with results regarding greedy/myopic sensing's performance bounds. 
Our own prior work \cite{Ahmad09optimalityof} established the optimality of the greedy policy for the $(1,1)$ case for arbitrary $N$ under the condition $p_{11}\geq p_{01}$, i.e., when a channel's state transitions are positively correlated.  This result was further generalized in \cite{Ahmad:2009:MOA:1793974.1794208} to the case of $(k,k)$, i.e., $m=k\leq N$ with arbitrary $N$.  

In view of the above existing work, the main contribution of this paper is the study of the more general $(k,m)$ problem with $1\leq m\leq k\leq N$.  For this problem we derive sufficient conditions under which the myopic policy is optimal for the finite and infinite horizon reward criteria, respectively, for both the positively correlated and negatively correlated channel models.  Furthermore, we show that they reduce to those derived in prior studies under the corresponding special cases, and thus may be viewed as a set of unifying optimality conditions.  
Our main results, a set of sufficient conditions for the optimality of the myopic policy, are summarized in Table \ref{tbl:summary}, where $0<\beta<1$ is the discount factor and $\overline{\cal R}$ and $\underline{\cal R}$ are two constants that depend on parameters $m$ and $k$.
\begin{table}[h!]
\begin{center}
    \begin{tabular}{ | c | c | c | } 
    \hline
~&~&~\\[0.2ex]
Channel model   & Finite horizon 
& Infinite horizon  \\ [2ex]\hline
~&~&~\\[0.2ex]
    $p_{11} \geq p_{01}$ & $\beta \leq \underline{\cal R}/\overline{\cal R}$  
    & $\frac{p_{11}-p_{01}}{1-(p_{11}-p_{01})} < \underline{\cal R}/\overline{\cal R}$ \\[2ex] \hline
~&~&~\\[0.2ex]
    $p_{11} < p_{01}$ & $\beta \leq \underline{\cal R}/(\underline{\cal R}+\overline{\cal R})$  
    & $ \min\{ p_{01}-p_{11}, \frac{1}{2(p_{00}+p_{11})} \} \leq \overline{\cal R}/\underline{\cal R}$ \\[2ex] \hline
    \end{tabular} 
\end{center}
\caption{Summary of results}
\label{tbl:summary}
\end{table}

\rev{The sufficient condition for the finite horizon problem is on $\beta$, and is derived using a sample path argument we first introduced in 
\cite{Ahmad:2009:MOA:1793974.1794208}. The sufficient condition for the infinite horizon problem is on $p_{11}$ and $p_{01}$, and is based a few bounding techniques and the one-step deviation principle.
It should be noted that similar results from a parallel development have recently appeared that address the case of positively correlated channels over a finite horizon for $m=1,k>1$ (in \cite{DBLP:journals/corr/abs-1104-5391}) and for $1<=m<=k$ (in \cite{DBLP:journals/corr/abs-1103-1784}), respectively.  They correspond to the upper left entry in Table \ref{tbl:summary}, and also rely on the sample path argument introduced in \cite{Ahmad:2009:MOA:1793974.1794208}. 
Paper \cite{DBLP:journals/tsp/WangC12} considers the additional relaxation to independent but non-identical channels (positively correlated and over a finite horizon). However, due to this generality the results obtained in \cite{DBLP:journals/tsp/WangC12} are weaker, i.e., their sufficient condition does not reduce to that in the special case of IID channels. 
By contrast, all sufficient conditions given in Table \ref{tbl:summary} reduce precisely to the best known results given in prior studies in respective special cases, thereby providing a unifying set of conditions. 

%
}

The remainder of this paper is organized as follows. Section \ref{sec:prob_formulation} presents the problem along with preliminary results. 
Sections \ref{sec:finite_positive} and \ref{sec:finite_negative} derive the optimality conditions for the finite horizon problem with positively and negatively correlated channels, respectively.  Sections \ref{sec:infinite_positive} and \ref{sec:infinite_negative} are similarly organized for the infinite horizon problem. 
Discussion and related work are given in Section \ref{sec:discussion} and Section \ref{sec:conclusion} concludes the paper.

\section{Problem Formulation and Preliminaries} \label{sec:prob_formulation}
\subsection{Problem formulation}


Denote the set of channels by $\mathcal N = \{1,2,...,N\}$. The system operates in discrete time $t=1,2,...$.  In each step $t$, the channel state transitions at $t^{-}$, followed by channel sensing at $t$.  The user is limited to sensing at most $k$ channels each time, thus its observation of the system when making decision at time $t$ is imperfect.  A sufficient statistic for optimal decision making, or the information state of the system \cite{kumar}, is given by the conditional probabilities of the state each channel is in given all past observations and actions.  Since each channel can be in one of two states, we denote this information state at time $t$ by 
%
$\bar{\omega}(t) := [\omega_1(t),\omega_2(t),...,\omega_N(t)]$, where $\omega_i(t)$ is the conditional probability that channel $i$ is in state $1$ at time $t$\footnote{Note that it is a standard way of turning a POMDP problem into a classic MDP problem by means of the information state, the main implication being that the state space is now uncountable.}.  
%
\rev{The user's sensing strategy is denoted by $\boldsymbol{\pi}^{1:T} = [\pi(1),\pi(2),...,\pi(T)]$, where $\pi(t): \omega(t)\rightarrow \Omega_k$, $\Omega_k\subset \Omega$ denoting a set of $k$ channels. $\pi(t)$ will be referred to as a policy, and $\Pi$ denotes the set of all admissible policies, while $\bar{\Pi}$ denotes the set of all admissible $T$-step policies.} 
Due to the Markovian nature of the channel model, future information state is only a function of the current information state and the current action.  It follows that the information state of the system evolves as follows.  Given $\bar{\omega}(t)$ and action $\pi(t)$, there are three possible state updates: (1) $\omega_i(t+1)=p_{11}$ if $i\in \pi(t)$ and it is observed in state $1$; (2) $\omega_i(t+1)=p_{01}$ if $i\in \pi(t)$ and it is observed in state $0$; (3) if $i \not\in \pi(t)$ then $\omega_i(t+1) = \tau(\omega_i(t))$, 
%
where $\tau(\cdot):[0,1] \rightarrow [0,1]$ is the updating function defined as
\begin{align}
\tau(\omega) = \omega \cdot p_{11} + (1-\omega) \cdot p_{01}, ~ 0\leq \omega\leq 1 ~. 
\end{align}

If a channel is sensed to be in state $1$ and the user decides to use it for transmission, then it gets a unit reward for that time step. 
%
The immediate one-step reward under state $\bar{\omega}$ and sensing \rev{action} $\pi$ is denoted by $R^{k,m}_{\pi}(\bar{\omega}), 1 \leq m \leq k$.  
\begin{ex}
The one-step reward of the $(k, 1)$ model (sensing $k\geq 1$ channels but using no more than one for data transmission) \rev{given policy $\pi \in \Pi$} is 
\begin{align}
\mathbb E [R^{k,1}_{\pi}(\bar{\omega})] = 1 - \prod_{i \in \pi}(1-\omega_i), 1\leq k \leq N~. \label{eqn:reward1}
\end{align}
\end{ex}
\begin{ex} 
The one-step reward of the $(k,k)$ model given \rev{$\pi\in\Pi$ is}  
\begin{align}
\mathbb E [R^{k,k}_{\pi}(\bar{\omega})] = \sum_{i \in \pi}\omega_i, 1\leq k \leq N ~. \label{eqn:reward2}
\end{align}
\end{ex} 

The objective for the finite horizon problem is to maximize the total expected discounted  reward over $T$ time steps, with a discount factor $0< \beta\leq 1$, given an initial state $\bar{\omega}$: 
\begin{align}
\textbf{(P1):}~ 
\rev{
\mathcal J^{\boldsymbol{\pi}}_T(\bar{\omega}) = \max_{\boldsymbol{\pi}\in\boldsymbol{\Pi}}~ \mathbb E^{\boldsymbol{\pi}}[\sum_{t=1}^T\beta^{t-1}R^{k,m}_{\pi(t)}(\bar{\omega}(t))|\mathbf{\omega}(1) = {\omega}] \nonumber
} 
\end{align}
The objective for the infinite horizon problem is to maximize the total expected discounted reward (with $0 < \beta < 1$) or the average reward: 
\begin{eqnarray} \textbf{(P2):}~ 
\mathcal J^{\boldsymbol{\pi}}_{\beta}(\bar{\omega}) &=& \max_{\boldsymbol{\pi}\in\boldsymbol{\Pi}}
~\mathbb E^{\boldsymbol{\pi}}[\sum_{t=1}^{\infty}\beta^{t-1}R^{k,m}_{\pi(t)}(\bar{\omega}(t))|\bar{\omega}(1) = \bar\omega] 
 \nonumber \\
 \textbf{(P3):}~ 
\mathcal J^{\boldsymbol{\pi}}_{\infty}(\bar{\omega}) &=& \max_{\boldsymbol{\pi}\in\boldsymbol{\Pi}}~ \mathbb E^{\boldsymbol{\pi}}[\lim_{T
\rightarrow \infty}\frac{1}{T}\cdot\sum_{t=1}^TR^{k,m}_{\pi(t)}(\bar{\omega}(t))|\bar{\omega}(1) = \bar\omega] \nonumber 
\end{eqnarray}

As we shall see a main technical challenge posed by the general $(k,m)$ problem is the non-additive nature of the reward function, see e.g., (\ref{eqn:reward1}), as opposed to the additive reward in the special case $(k,k)$ as shown in (\ref{eqn:reward2}), in addition to the usual difficulties in seeking structural solutions to restless bandit problems. 
As in previous works, we will focus on a simple myopic policy that aims at maximizing the immediate, one-step reward at each time step, and investigate under what conditions this policy is optimal.  In the remainder of this section we present a number of properties of the above non-additive reward function and the operation of the myopic policy in the context of the dynamic programming representation of the above optimization problems.  


\subsection{Properties of the expected reward $\mathbb E [R^{k,m}_{\pi}(\bar{\omega})]$}\label{positive}

For convenience of notation, the vector $\mathbf{\omega}$ will be frequently written as $(\omega_i, \omega_{-i})$ to emphasize the $i$-th element and the rest of the vector, or as $(\omega_1, \cdots, \omega_i, \cdots, \omega_N)$.  The first property below suggests that the order in which these elements appear does not matter.  For this reason later we will sort them in descending order. 

\begin{proposition}[Symmetric]
Under any admissible policy $\pi$, $\forall i,j \in \cal N$ and $\omega_i = \omega_j$ we have
\begin{align}
\mathbb E [R^{k,m}_{\pi}(\omega_1,...,\omega_i,...,\omega_j,...,\omega_N)] = \mathbb E [R^{k,m}_{\pi}(\omega_1,...,\omega_j,...,\omega_i,...,\omega_N)] ~. 
\end{align}
\end{proposition}
The above property is quite self-evident and its proof is thus omitted. 

\begin{proposition}[Increasing]\label{eqn:prop_increase}
For $\omega^{'}_i > \omega_i$ we have
\begin{align}
\mathbb E [R^{k,m}_{\pi}(\omega^{'}_i,\mathbf{\omega_{-i}})] \geq \mathbb E [R^{k,m}_{\pi}(\omega_i,\mathbf{\omega_{-i}})] ~. 
\end{align}
\end{proposition}
\begin{proof} 
If $i\not\in \pi$, then the two sides must be equal because all other elements are the same. 
Consider the case $i\in \pi$.  The immediate one-step reward can be expressed in the following sequential form: 
\begin{align}
E [R^{k,m}_{\pi}(\omega_i,\mathbf{\omega_{-i}})] = \omega_i\cdot (E [R^{k-1,m-1}_{\pi_{-i}}(\mathbf{\omega_{-i}})]+1) + (1-\omega_i)E [R^{k-1,m}_{\pi_{-i}}(\mathbf{\omega_{-i}})] ~, \label{eqn:recur}
\end{align}
where $\pi_{-i}$ denotes the same set of channels in $\pi$ but excluding $i$. 
This is because since all available channels generate the same reward, we may consider two possibilities of obtaining the total reward: either channel $i$ is available or not. Under the former, we receive the unit reward plus the reward from the remaining $k-1$ channels in $\pi$, using up to $m-1$ of them; under the latter, the total reward now comes from the remaining $k-1$ channels in $\pi$, using up to $m$ of them. 
Applying (\ref{eqn:recur}) to both sides of (\ref{eqn:prop_increase}), in order to show the inequality in (\ref{eqn:prop_increase}) it suffices to show that 
\begin{align}
E [R^{k-1,m-1}_{\pi_{-i}}(\mathbf{\omega_{-i}})]+1 > E [R^{k-1,m}_{\pi_{-i}}(\mathbf{\omega_{-i}})] ~. 
\end{align} 
Next we show this is true. Let $P_{\pi_{-i}}(l)$ denote the probability that out of $k-1$ channels in $\pi_{-i}$, exactly $l$ are sensed to be good under state $\omega$. We have 
\begin{align}
E [R^{k-1,m-1}_{\pi_{-i}}(\mathbf{\omega_{-i}})]+1 &= \sum_{l=0}^{m-2}P_{\pi_{-i}}(l)\cdot l + \sum_{l=m-1}^{k-1} P_{\pi_{-i}}(l)\cdot (m-1) + 1\nonumber \\
&>   \sum_{l=0}^{m-2}P_{\pi_{-i}}(l)\cdot l +P_{\pi_{-i}}(m-1)\cdot (m-1) +  \sum_{l=m}^{k-1} P_{\pi_{-i}}(l)\cdot [(m-1) + 1] \nonumber \\
&=  \sum_{l=0}^{m-2}P_{\pi_{-i}}(l)\cdot l +P_{\pi_{-i}}(m-1)\cdot (m-1) +  \sum_{l=m}^{k-1} P_{\pi_{-i}}(l)\cdot m \nonumber \\
&=  E [R^{k-1,m}_{\pi_{-i}}(\mathbf{\omega_{-i}})] ~. 
\end{align} 
\end{proof} 

The fact that (\ref{eqn:recur}) is an affine function of $\omega_i$ also leads to the next result. 

\begin{proposition}[Affine]
$\mathbb E[R^{(k,m)}_{\pi}(\bar{\omega})]$ is an affine function w.r.t. each $\omega_i, \forall i \in \pi$, i.e.,
\begin{align}
&\mathbb E [R^{k,m}_{\pi}(\omega_i = x,\mathbf{\omega_{-i}})] - \mathbb E [R^{k,m}_{\pi}(\omega_i = y,\mathbf{\omega_{-i}})] \nonumber \\
&= (x-y) \cdot \{\mathbb E [R^{k,m}_{\pi}(\omega_i = 1,\mathbf{\omega_{-i}})] - \mathbb E[R^{k,m}_{\pi}(\omega_i = 0,\mathbf{\omega_{-i}})]\}\label{prop_affine}
\end{align}
\end{proposition}

\subsection{Dynamic programming representation}

Throughout this paper we will consider the general $(k,m)$ case, and for simplicity will use $R_{\pi}(\bar{\omega})$ thereafter instead of $R^{k,m}_{\pi}(\bar{\omega})$ whenever there is no confusion. 
The optimization problem (P1) can be solved using dynamic programming:  
\begin{align}
V_T(\bar{\omega}) &= \max_{\pi \in \Pi} ~\mathbb E [R_{\pi}(\bar{\omega})] ~,\\
V_t(\bar{\omega}) 
&=\max_{\pi\in \Pi}~\mathbb E [R_{\pi}(\bar{\omega})] 
+ \beta \cdot \sum_{l_i \in \{0,1\}, i \in \pi} \prod_{i \in \pi}(\omega^{l_i}_i(1-\omega_i)^{1-l_i}) \nonumber \\
&\cdot V_{t+1}(p_{11}[\sum_{i \in \pi} l_i],\tau({\omega_j}),..,p_{01}[k-\sum_{i \in \pi} l_i])~,  \label{multi_vp}
\end{align}
where we have adopt the following notation for simplicity: 
\begin{itemize}
\item $p_{01}[x]$: a vector $[p_{01},p_{01},...,p_{01}]$ of length $x$.
\item $p_{11}[x]$: a vector $[p_{11},p_{11},...,p_{11}]$ of length $x$.
\end{itemize}
In (\ref{multi_vp}), the state vector in $V_{t+1}(\cdot)$ consists of three parts: channels in $\pi$ and sensed to be good (their next state is $p_{11}$); channels in $\pi$ and sensed to be bad (their next state is $p_{01}$); and channels not sensed (their next state is $\tau({\omega_j})$). 
%

\subsection{The myopic/greedy sensing policy}

The myopic/greedy sensing policy selects a set of channels so as to maximize the one-step immediate reward. 
\rev{If we sort an information state $\bar{\omega}(t)$ in descending order such that $\omega_1(t) \geq \omega_2(t) \geq ... \geq \omega_N(t)$, then myopic sensing, denoted by $\pi^g$, is one that selects the first $k$ channels (highest probabilities of being good), i.e,
$\pi^g= \{1,2,...,k\}$ for a descending ordered $\omega$. Note however $\pi^g$ can be applied to an arbitrarily ordered $\omega$; it will simply selects the first $k$ channels.} 
%
As detailed in \cite{Zhao&etal:08TWC,Ahmad:2009:MOA:1793974.1794208} the implementation of the myopic strategy is particularly simple: it only requires the knowledge of the ordering of the initial information state and the ordering of $\{p_{11},p_{01}\}$. 
Since this feature is repeatedly used in our analysis, below we elaborate on this to make the paper self-contained.

For the case when $p_{11} \geq p_{01}$, the updating function $\tau(\omega)$ is monotonically non-decreasing, i.e., $\tau(\omega_1) \geq \tau(\omega_2)$ if $\omega_1 \geq \omega_2$, implying that the ordering of channels not sensed is preserved.  The states of sensed channels are updated to either $p_{11}$ (if sensed good) or $p_{01}$ (if sensed bad), noting that $p_{01} \leq \tau(x) \leq p_{11}, \forall x \in  [0,1]$.  It follows that we have the following simple implementation of the myopic policy:
Starting from a descending-ordered list of channels, the policy selects the first $k$ channels.  Upon learning the sensing outcome, those sensed to be good are placed at the front of the list, those sensed to be bad at the end of the list, and those not sensed are in the middle in their original order.  By the above observation, this new list is again in descending order, and thus the policy again selects the first $k$ channels for the next time step, and the same process is repeated. 

For the case with $p_{11} < p_{01}$ we also have monotonicity but in the opposite direction, i.e., $\tau(\omega_1) \geq \tau(\omega_2)$ if $\omega_1 \leq \omega_2$.  Thus the ordering those not sensed is  reversed at each time step.  Meanwhile $p_{11} \leq \tau(x) \leq p_{01}, \forall x \in  [0,1]$. 
A similar implementation thus follows: at each time step we place the channels sensed as good to the end  of the list, those sensed bad at the front of the list, and those not sensed in the middle with their ordering reversed.  This produces a descending ordered list so that at the next time step the policy again selects the first $k$ channels.

While both the expected one-step reward and the value functions are invariant w.r.t. the ordering of the information state/belief vector $\omega$, for simplicity of presentation we will take $\omega$ to be an ordered vector for the remainder of this paper.  Accordingly, the notation $(\omega_i, \omega_{-i})$ is used to represent the following ordered vector: $(\omega_i, \omega_1, \cdots, \omega_{i-1}, \omega_{i+1}, \cdots, \omega_N)$. 


\section{
Finite Horizon, $p_{11} \geq p_{01}$} \label{sec:finite_positive}

\subsection{Optimality of myopic sensing}

We begin by introducing the following two quantities: 
\begin{align}
\overline{\mathcal R} = \max_{\omega_{-i} \in [p_{01},p_{11}]^{k-1}}\{\mathbb E [R_{\pi^g}(1,\omega_{-i})] -  \mathbb E [R_{\pi^g}(0,\omega_{-i})]\}\\
\underline{\mathcal R} = \min_{\omega_{-i} \in [p_{01},p_{11}]^{k-1}}\{\mathbb E [R_{\pi^g}(1,\omega_{-i})] -  \mathbb E [R_{\pi^g}(0,\omega_{-i})]\} ~. 
\end{align}
$\overline{\mathcal R},\underline{\mathcal R}$ can be easily characterized for some commonly used cases; some examples are shown below. 

\begin{ex}{$(k,m) = (k,k), 1 \leq k \leq N$}\label{ex_kk}
In this case we can sense up to $k$ channels and use all those sensed to be available. The one-step reward under $\pi^g$ is thus  
$\mathbb E [R_{\pi^g}(\bar{\omega})] = \sum_{i \in \pi^g} \omega_i = \sum_{i=1}^{k} \omega_i$, 
and thus
$\overline{\mathcal R} = \underline{\mathcal R} = 1$. 
\end{ex}

\begin{ex}{$(k,m) = (k,1)$.} \label{ex_k1}
Since we can use no more than 1 channel, the one-step reward under $\pi^g$ is given by 
$\mathbb E [R_{\pi^g}(\bar{\omega})] = 1 - \prod_{i=1}^k (1 - \omega_i)$, 
and thus
$\overline{\mathcal R} = (1-p_{01})^{k-1}, \underline{\mathcal R} = (1-p_{11})^{k-1}$. 
\end{ex} 

We now present the main result of this section. 
\begin{theorem}[Optimality of Myopic Sensing] \label{thm:finite_positive}
The myopic sensing policy $\pi^g$ 
is optimal for $\textbf{(P1)}$ under the condition $0 \leq \beta \leq 
\underline{\cal R}/\overline{\cal R}$ \rev{and for belief state $\bar{\omega}$ s. t. $p_{01}\leq \omega_i \leq p_{11}, \forall \omega_i\in \bar{\omega}$}. 
\end{theorem}

\rev{
\begin{remark}
Note that the condition on $\bar{\omega}$ in the above theorem is not overly restrictive, as $p_{01}\leq \tau(\omega_i)\leq p_{11}$ for any $\omega_i$, implying that even if the initial belief $\bar{\omega}$ at time $t=1$ does not satisfy this condition, the theorem is applicable starting from time $t=2$. 
\end{remark} 
} 

To prove this theorem, we next introduce a number of lemmas.  Define $T$ $N$-variable functions $W_t(\cdot), t=1, 2, \cdots, T$, recursively as follows:   
\begin{eqnarray}
\rev{W_T(\bar{\omega}) } &=& \rev{\mathbb E [R_{\pi^g}(\bar{\omega})]} \nonumber \\ 
 W_t(\bar{\omega}) 
&=&  \mathbb E [R_{\pi^g}(\bar{\omega})] + \nonumber \\
&& \beta \cdot \sum_{\bar{l} \in \{0,1\}^k} 
q(\bar{l}; \bar{\omega}) \cdot 
W_{t+1}(p_{11}[\sum_{i=1}^k l_i],\tau({\omega_{k+1}}),.., \tau(\omega_{N}), p_{01}[k-\sum_{i=1}^k l_i]) , \label{eqn:W_def}
\end{eqnarray}
\rev{where $\bar{l}=\{l_1, \cdots, l_k\}$, and $q(\bar{l}; \bar{\omega}) := \prod_{i =1}^{k}(\omega^{l_i}_i(1-\omega_i)^{1-l_i}), l_1,l_2,...,l_k \in \{0,1\}$.} 


\begin{remark}
A few remarks are in order on these function $W_t(\cdot),t=1,2,...,T$: 
\begin{enumerate}
\rev{
\item  If $\bar{\omega}$ is in descending order, then applying $\pi^g$ at time $t$ is myopic.  Moreover, the state vector within $W_{t+1}(\cdot)$ retains the same descending order.  This is because $\tau(\omega)$ is increasing in $\omega$ and $p_{11} \geq \tau(\omega) \geq p_{01}$ for any $\omega$.  Thus if $\omega_{k+1}\geq \cdots \geq \omega_N$, then $p_{11}\geq \omega_{k+1} \geq \cdots \geq \omega_N \geq p_{01}$.  This implies that selecting the first $k$ channels at $t+1$, i.e., $\pi^g$ would again be myopic. 

\item When $\bar{\omega}$ is in descending order of its components, $W_t(\bar{\omega})$ is the expected discounted total reward starting from state $\bar{\omega}(t)$ at time $t$ by following the myopic policy at each time step.  This is because $W_t(\cdot)$ takes on the same recursive form as the value function, and at each time step the myopic policy is used due to the descending order of the state vector as noted above.  

\item When $\bar{\omega}$ is not in descending order, $W_t(\bar{\omega})$ as given above represents the expected discounted total reward of the following policy: It selects the first $k$ channels as listed in the vector $\bar{\omega}$ at time $t$; it then orders the next state vector as follows: those channel sensed to be good are listed first, followed by those not sensed at all, in their original order in $\bar{\omega}$, followed finally by those sensed to be bad.  This process is then repeated. 

\item When $j\in \pi^g$ ($1\leq j\leq k$), we can also conveniently write $W_t(\bar{\omega})$ in the following form by singling out component $\omega_j$ and calculating the expected future reward conditioned on the outcome of sensing channel $j$; this expression is frequently used in our proofs: 
\begin{eqnarray}
 W_t(\bar{\omega}) 
&=&  \mathbb E [R_{\pi^g}(\omega_j, \omega_{-j})] + \nonumber \\
&& \omega_j \beta \cdot \sum_{\bar{l}_{-j} \in \{0,1\}^{k-1}} 
q(\bar{l}_{-j}; \omega_{-j}) 
W_{t+1}(p_{11}[\sum_{i\neq j} l_i+1],\tau({\omega_{k+1}}),.., \tau(\omega_{N}), p_{01}[k-\sum_{i\neq j} l_i-1])+  \nonumber\\
&& (1-\omega_j) \beta \cdot \sum_{\bar{l}_{-j} \in \{0,1\}^{k-1}} 
q(\bar{l}_{-j}; \omega_{-j})  
W_{t+1}(p_{11}[\sum_{i\neq j} l_i],\tau({\omega_{k+1}}),.., \tau(\omega_{N}), p_{01}[k-\sum_{i\neq j} l_i]), \label{eqn:W_alt}
\end{eqnarray}
where $\bar{l}_{-j} = \{l_1, \cdots, l_{j-1}, l_{j+1}, \cdots, l_k\}$. } 
\end{enumerate}
\end{remark}

Key properties of the functions $W_t(\cdot), t=1, 2, \cdots, T$ are presented below. 
\begin{lemma}[Monotonicity] 
$W_t(\bar{\omega}^{'}) \geq W_t(\bar{\omega})$, $t=1,2,...,T$, for $\bar{\omega}^{'}~ \mathbf{\succeq}~\bar{\omega}$, with $\mathbf{\succeq}$ denoting  component wise larger than or equal to. \label{lemma:mono}
\end{lemma}

\begin{lemma}[Affine]
$W_t(\bar{\omega}), t=1, 2, \cdots, T$, is an affine function of each element of $\bar\omega$. \label{lemma:affine}
\end{lemma}
\begin{proof} 
\rev{We prove this by induction on $t$.  Consider $W_T(\bar{\omega})$ and an element $\omega_j$.  If $j\notin \pi^g$, then $W_T(\bar{\omega})$ is not a function of $\omega_j$.  If $j \in \pi^g$, then $\mathbb E [R_{\pi^g}(\bar{\omega})]$ is an affine function of $\omega_j$ by Proposition \ref{prop_affine}.  In either case the induction basis is established. 
Suppose the lemma holds for all times $t+1, t+2, \cdots, T$.  

Now consider $W_t(\bar{\omega})$, and the case $j\notin\pi^g$.  
%
%
By the induction hypothesis, the $W_{t+1}(\cdot)$ term in (\ref{eqn:W_def}) is an affine function of $\tau(\omega_j)$, which in turn is a linear in $\omega_j$.  Since $W_t(\bar{\omega})$ only depends on $\omega_j$ through this $W_{t+1}(\cdot)$ function, by the definition in (\ref{eqn:W_def}), it follows that $W_t(\bar{\omega})$ is affine in $\omega_j$.  

Consider the case $j\in \pi^g$. 
In this case $\mathbb E [R_{\pi^g}(\bar{\omega})]$ and $q(\bar{l};\bar{\omega})$ are both affine functions of $\omega_j$ (by Proposition \ref{prop_affine} and definition of $q(\cdot)$, respectively).  Meanwhile the $W_{t+1}(\cdot)$ term in (\ref{eqn:W_def}) does not depend on $\omega_j$ as $j\in \pi^g$.  Thus $W_t(\bar{\omega})$ is again affine in $\omega_i$. 
} 
\end{proof}

The next lemma provides two key inequalities that lead to the proof of the main theorem in this section. 

\begin{lemma}\label{lem:key}
For $p_{11} \geq \omega_1 \geq \omega_2 \geq ... \geq \omega_N \geq p_{01}$ and for all $t = 1, 2, \cdots, T$\footnote{The assumption of bounding $\bar{\omega}$ between $p_{01}$ and $p_{11}$ is in fact a rather weak one. To see this it is easy to verify $p_{01} \leq \tau(x) \leq p_{11}, \forall x \in [0,1]$; thus if the initial belief falls between $[p_{01},p_{11}]$ (for example taking the initial belief as the steady state distribution $\frac{p_{01}}{p_{01}+p_{11}},\frac{p_{11}}{p_{01}+p_{11}}$), the assumption holds immediately for any $t$. } , under the condition $\beta \leq \underline{\cal R}/\overline{\cal R}$ and $ x,y$ we have: \com{do $x$ and $y$ need to be in the range of $p_{01}$ and $p_{11}$?} \com{we donnot have the constraint for x,y.}
\begin{align}
\mbox{(L1):}~~~& \overline{\mathcal R} + W_t(\omega_N,\omega_1,...,\omega_{N-1}) \geq W_t(\omega_1,...,\omega_N) ~, \\
\mbox{(L2):}~~~& W_t(\omega_1,...,\omega_{j-1}, x, y, \omega_{j+2}, \cdots,\omega_N) \geq W_t(\omega_1, \cdots, \omega_{j-1}, y, x, \cdots, \omega_{j+2}, \cdots, \omega_N) ~. 
\end{align}
\end{lemma}
%

{\em Proof of Theorem \ref{thm:finite_positive}:} 
\rev{We prove the theorem by induction on $t$.  

{\em Induction basis:} That $\pi^g$ is optimal at time $T$ is obvious due to the increasing property of the expected one-step reward, Proposition \ref{eqn:prop_increase}. Assume the myopic policy $\pi^g$ is optimal for any given state vector $\omega$ for times $t+1, \cdots, T$. 

{\em Induction step:} Suppose the optimal policy at time $t$ under state $\bar{\omega}$ is $\pi^{*} \neq \pi^g$. Accordingly, we can write the state vector as $(\bar{\omega}_{*}, \bar{\omega}_{-*})$, where $\bar{\omega}_{*} :=\{ \omega_j, j\in \pi^{*} \}$ contains the probabilities of those channels selected by $\pi^{*}$ and $\bar{\omega}_{-*}:=\bar{\omega}-\bar{\omega}_{*}$, sorted in descending order, contains those not selected by $\pi^{*}$.  Since the myopic policy is optimal starting from $t+1$ by the induction hypothesis, the expected discounted reward of using policy $\pi^{*}$ at time $t$ followed by the myopic policy thereafter is essentially given by $V_t^{\pi^{*}}(\bar{\omega}) = W_t(\bar{\omega}_{*}, \bar{\omega}_{-*})$, where $\bar{\omega}$ is in descending order. However, by repeated use of L2 in Lemma \ref{lem:key}, sorting one element at a time, we have $W_t(\bar{\omega}) \geq W_t(\bar{\omega}_{*}, \bar{\omega}_{-*})$, contradicting the claim.  Therefore the myopic policy is also optimal at time $t$. \qed
} 


\subsection{Special cases}\label{sec:finite_positive_special}

We next interpret the result obtained above in a number of special cases. 

\textbf{Case 1: $(k, m)=(k, k)$.} 
As shown earlier in Example \ref{ex_kk} we have
$\overline{\mathcal R} = \underline{\mathcal R}  = 1$.  Thus in this case the optimality condition reduces to $\beta\leq 1$ which is always true, i.e., it is not binding. 

\textbf{Case 2: $(k, m) = (k, 1)$.} 
As shown earlier in Example \ref{ex_k1} we have 
$\overline{\mathcal R} = (1-p_{01})^{k-1}$ and $\underline{\mathcal R} = (1-p_{11})^{k-1}$. 
It follows that $\underline{\cal R}/\overline{\cal R}<1$, except for the trivial case of $p_{11} = p_{01}$.  This means that in the case of sensing multiple channels while limiting access to one channel, the myopic policy is not always optimal, and the optimality condition $\beta\leq \underline{\cal R}/\overline{\cal R}$ becomes binding.

\textbf{Case 3: $k=N, m\leq k$.}
This case is trivial as only a single action is available at each time, which coincides with the myopic policy when $k=N$.  It is therefore optimal without requiring any conditions. 

\textbf{Case 4: $k=N-1, m\leq k$.}
It can be shown that in this case the myopic policy is optimal without any condition on $\beta$ or $\bar{\omega}$. \com{??? true??}\com{yang: true. under this case, no constraint for $\beta$ or $\omega$ will be needed.} 
The proof follows the same argument used in the preceding subsection.  In particular, we note that the 
condition on $\beta$ arise from the induction step of proving L2 in Lemma \ref{lem:key}.  However, it can be easily verified that when $k=N-1$ this step holds for all $0\leq \beta\leq 1$. 

\subsection{A numerical example}

\rev{The following numerical example highlights how myopic sensing my not be optimal when the sufficient condition on $\beta$ is not satisfied. } \com{Is this what the example shows? or does it simply show what we have now remarked, which that for beliefs outside the range myopic sensing is not necessarily optimal?} 

\com{yang: under this example, both constraints are violated. $\bar{\omega}$ is out of the range of $[p_{01},p_{11}]$ and $\beta$ is also out of the bounding range.}
The example is given by the following parameter values: 
$N=5, k=2, m=1, \beta=0.8, T=5, p_{11}=0.9, p_{01}=0.1$, 
%
with an initial information states 
$\bar{\omega} = \{0.99,0.95,0.9,0.9,0.9\}$. 
%
Denote by $W^{\{1,2\}}_1$ the expected reward of sensing myopically (channels ordered \{1, 2\}) in each time step, and by $W^{\{1,3\}}_1$ the expected reward of sensing channels  \{1, 3\} at $t=1$ followed by sensing myopically thereafter.  Numerically solving the example shows that 
$W^{\{1,2\}}_1 = 3.3279$ and 
$W^{\{1,3\}}_1 = 3.3283$,  
thus in this case myopic sensing is not optimal.  \com{again, what is the reason: $\beta$ or $\bar{\omega}$?}  

\rev{
What this counter example shows is that when the top channel (the one with highest information state) has a sufficiently high belief, i.e. we have high confidence that in the next step this channel will be available, it may make more sense to take this opportunity to {\em explore} by updating our belief on a lower channel (number 3 in this case) rather than selecting the second highest channel to further improve our chance (which is already very high by virtue of the top channel's state) of getting at least one good channel in the next time step. 
} 

It is worth noting that these counter examples are only found in such extreme cases, i.e., cases with information state close to 1, or cases with high $p_{11}$ and low $p_{01}$. \com{need to elaborate on this a bit more} 
\section{Finite Horizon, $p_{11} < p_{01}$} \label{sec:finite_negative}

\subsection{Optimality of myopic sensing}
\begin{theorem}[Optimality of Myopic Sensing] \label{thm:finite_negative} 
The myopic sensing strategy $\pi_g$ 
is optimal for $\textbf{(P1)}$ under the condition $0 \leq \beta \leq \frac{\underline{\mathcal R}}{\underline{\mathcal R}+\overline{\mathcal R}}$ \rev{and for belief state $\bar\omega$ s. t. $p_{11}\leq \omega_i \leq p_{01}, \forall \omega_i\in \bar{\omega}$}. 
\end{theorem}

\rev{We will reuse the same set of notations introduced in the case of $p_{11} \geq p_{01}$ in this section.  To prove the above theorem, we will similarly need a number of lemmas. 
We begin with a similar definition on the $T$ $N$-variable functions $W_t(\cdot), t = 1, 2, \cdots, T$, recursively as follows. 
\begin{eqnarray}
\rev{W_T(\bar{\omega}) } &=& \rev{\mathbb E [R_{\pi^g}(\bar{\omega})]} \nonumber \\ 
 W_t(\bar{\omega}) 
&=& \mathbb E [R_{\pi^g}(\bar{\omega})] + \nonumber \\
&& \beta \cdot \sum_{\bar{l} \in \{0,1\}^k}q(\bar{l};\bar{\omega})  
 \cdot W_{t+1}(p_{01}[k-\sum_{i=1}^k l_i],\tau({\omega_N}), \cdots, \tau(\omega_{k+1}), p_{11}[\sum_{i=1}^k l_i]) \label{lemma1_neg}
\end{eqnarray}
\begin{remark}
Compared to the definition given in the previous section, the difference here is in the re-ordering of the beliefs in $W_{t+1}(\cdot)$, i.e., $p_{01}$'s followed by $\tau(\omega_N), \cdots$, followed by $p_{11}$'s.  This is because, as $p_{01}> p_{11}$, this re-ordering sorts the belief vector in descending order.  In doing so we can continue to use the same greedy policy $\pi^g$ which selects the first $k$ channels. 
\end{remark} 
} 

\begin{lemma}
$W_t(\bar{\omega}), t=1, 2, \cdots, T$, is an affine function of each element of $\bar{\omega}$. \label{lemma2_neg}
\end{lemma}
The proofs of the above lemma is essentially the same as that in the case of $p_{11} \geq p_{01}$ (Lemma \ref{lemma:affine}), and is thus omitted.

\begin{lemma}\label{lem:key2}
For $p_{01} \geq \omega_1 \geq \omega_2 \geq ... \geq \omega_N \geq p_{11}$ \rev{and under the condition $\beta\leq \frac{\underline{\cal R}}{\underline{\cal R}+\overline{\cal R}}$ and $x \geq y$}, we have the following inequalities for all $t = 1,2,...,T$:
\begin{align}
\mbox{(L3):}~~~& \gamma + W_t(\omega_2,\omega_3,...,\omega_{N},\omega_{1}) \geq W_t(\omega_1,...,\omega_N)\\
\mbox{(L4):}~~~&\gamma + W_t(\omega_N,\omega_1,...,\omega_{N-1}) \geq W_t(\omega_1,...,\omega_N) \\
\mbox{(L5):}~~~& W_t(\omega_1, \cdots, \omega_{j-1}, x, y, \omega_{j+2}, \cdots, \omega_N) \geq W_t(\omega_1, \cdots, \omega_{j-1}, y, x, \omega_{j+2}, \cdots, \omega_N) ~, 
\end{align}
where $\gamma = \frac{\overline{\mathcal R}}{1-\beta}$. 
%
\end{lemma}

\rev{
\emph{Proof of Theorem \ref{thm:finite_negative}:} The proof follows essentially the same inductive argument used in the proof of Theorem \ref{thm:finite_positive} through repeated use of L5 in the preceding lemma, and is thus omitted. 
\qed} 

\subsection{Special cases}

{\bf Case 1:} $(k, m) = (k, k)$.  As shown earlier, in this case we have 
$\overline{\mathcal R} =  \underline{\mathcal R} = 1$, and thus the sufficient condition for the optimality of the myopic policy becomes
$\beta \leq \frac{\underline{\mathcal R}}{\underline{\mathcal R}+\overline{\mathcal R}} = \frac{1}{2}$.  Note that the same condition $\beta \leq \frac{1}{2}$ was previously proven for the special case $k=1$ in \cite{Ahmad09optimalityof}. 

{\bf Case 2:} $k = N-1,  m\leq k$. It can be shown in this case that the myopic policy is optimal without any condition on $\beta$ \rev{and $\bar{\omega}$} following the same argument used in Section \ref{sec:finite_positive_special}. 

\section{Infinite Horizon: $p_{11} \geq p_{01}$}\label{sec:infinite_positive}

In this and the next sections we will consider the infinite horizon problems \textbf{(P2)} and  \textbf{(P3)}. 
As shown in \cite{Ahmad09optimalityof}, the optimality of a policy under {\bf (P1)} is readily extended to its optimality under \textbf{(P2)}; it is more complicated for \textbf{(P3)}: a policy is optimal for \textbf{(P3)} if it is optimal for \textbf{(P2)} for any $0 < \beta < 1$\footnote{In \cite{Ahmad09optimalityof} this argument is made specifically for the case $(k, m) = (1, 1)$, but it is more generally applicable with a  simple extension.}. 
As a result, while the optimality conditions on the myopic policy we have obtained so far applies to \textbf{(P2)}, the same cannot be said for \textbf{(P3)} since these conditions restrict the values the discount factor $\beta$ can take. 
\rev{For this reason, in these two sections we seek alternative sufficient conditions that do not require the restriction on $\beta$, which will then allow us to first establish the optimality of the myopic policy for \textbf{(P2)} and then extend it to \textbf{(P3)}. } 

\subsection{One-step deviation}
\rev{For the rest of this section we will use the notation $W^{\infty}(\bar{\omega})$ defined similarly as in (\ref{eqn:W_def}) for the case of $p_{11} \geq p_{01}$ but with an infinite horizon, i.e., with the recursion in (\ref{eqn:W_def}) continuing indefinitely without the end at time $T$.} To be specific we have the following recursive equations.
\begin{eqnarray}
W^{\infty}_t(\bar{\omega})  &=& \mathbb E [R_{\pi^g}(\bar{\omega})] 
 + \beta \cdot \sum \textbf{P}^{g}(\bar{\omega}^{'}|\bar{\omega})\cdot W^{\infty}_{t+1}(\bar{\omega}^{'}) ~,  
\end{eqnarray}
But notice here the real value of the value functions does not depend on time $t$ due to the infinite horizon. We keep the time index mainly for clarity of later analysis. 
\begin{defn}[One-step deviation]
\rev{
Consider a policy $\pi^d: \omega\rightarrow \Omega_k$, $\pi^d\neq \pi^g$. 
Its one-step deviation from the myopic policy under information state $\omega$ is defined as the immediate reward under $\pi^d$ plus the discounted future reward by following $\pi^g$ in future time steps.  
Formally, the value function of $\pi^d$, denoted by $V^{d,\infty}_t(\omega)$, is given by 
\begin{eqnarray}
V^{d,\infty}_t(\bar{\omega})  &=& \mathbb E [R_{\pi^d}(\bar{\omega})] 
 + \beta \cdot \sum \textbf{P}^{d}(\bar{\omega}^{'}|\bar{\omega})\cdot W^{\infty}_t(\bar{\omega}^{'}) ~,  
\end{eqnarray}
where $\omega^{'}$ is the descending-ordered information state vector of the system at the next time step  under policy $\pi^d$.  If $V^{d,\infty}_t(\omega) > W^{\infty}_t(\omega)$ for some $\omega$ and $t$, then we say that $\pi^d$ is a profitable one-step deviation.  If such a $\pi^d$ cannot be found, then we say there exists no profitable one-step deviation. 
} 
\end{defn}

\begin{lemma} [One-step deviation principle]
\label{lem:one-step} 
The myopic
policy $\pi^g$ is optimal for \textbf{(P2)} for any $0 < \beta <1$ if and only if there exists no profitable one-step deviation.
\end{lemma}

\emph{Proof.}  (Only if) That there is no one-step profitable deviation is a necessary condition for the optimality of $\pi^g$ is obvious because otherwise we have found a policy that returns higher reward than $\pi^g$ under some state $\omega$, which contradicts the optimality of $\pi^g$. 

{
(If) We next show that if there exists a policy $\pi^* :\Omega\rightarrow{\cal N}_k$ that has strictly higher discounted reward than $\pi^g$ over an infinite horizon, then there exists a one-step profitable deviation policy constructed from $\pi^*$. Denote the total reward under $\pi^{*}$ starting at time $t$ as $V^{*, \infty}_t$, and denote by $\epsilon = V^{*,\infty}_1 - W^{\infty}_1$.  
By assumption we have $\epsilon  > 0$. Define time $t^*$ as
\begin{align}
t^* := \min \{t: \beta^t\cdot \frac{m}{1-\beta} \leq \frac{\epsilon}{2} \} ~, 
\end{align}
i.e., this is the first time that the total future discounted reward of an {\em ideal} policy (that collects the highest reward $m$ in each step) falls below $\epsilon/2$. The existence of such a $t^*$ is guaranteed by the finiteness of $m$ and the fact that $\beta < 1$.
By the above definition, after time $t^*$ the reward under either $\pi^*$ or $\pi^g$ cannot exceed $\epsilon/2$, thus the difference in the two rewards after time $t^*$ is no more than $\epsilon/2$.  Since the total difference between the two rewards (starting at time $t=1$) is $\epsilon$, the difference between $\pi^*$ and $\pi^g$ up to and including time $t^*$ must be at least  $\epsilon/2$.  We thus construct the following policy, $\pi^{+}$, which follows $\pi^{*}$ up to and including time $t^*$, and then switch to $\pi^g$ thereafter, with a total discounted reward denoted by $V^{+,\infty}_1(\cdot)$.   
Following the above discussion, we must have 
$V_1^{+,\infty}(\bar{\omega}) > W^{\infty}_1(\bar{\omega})$ for any initial condition $\omega$. 
} 

\rev{
Consider now the policy $\pi^{+}$.  At time $t^*$ we compare $V_{t^{*}}^{+,\infty}(\bar{\omega})$ with $W_{t^{*}}^{\infty}(\bar{\omega})$, $\forall \bar\omega$.  Note that in this case $V_{t^{*}}^{+,\infty}(\bar{\omega})=V_{t^{*}}^{*, \infty}(\bar{\omega})$ since under $\pi^{+}$ at time $t^*$ policy $\pi^*$ is used followed by $\pi^g$. If $V_{t^{*}}^{+,\infty}(\bar{\omega}) > W_{t^{*}}^{\infty}(\bar{\omega})$ for some $\bar\omega$, then we have found a profitable one-step deviation.  If $V_{t^{*}}^{+,\infty}(\omega) \leq W_{t^{*}}^{\infty}(\omega)$, $\forall\omega$, then we modify policy $\pi^{+}$ by replacing $\pi^{*}$ with $\pi^g$ at time $t^{*}$.  Again denote this modified policy by $\pi^{+}$; it follows that we continue to have $V_1^{+,\infty}(\bar{\omega}) > W_1^{\infty}(\bar{\omega})$ for any initial condition $\bar{\omega}$, since the modified $\pi^{+}$ has even higher total discounted rewards than the original $\pi^{+}$.  

We next examine at $t^{*}-1$, how $V_{t^{*}-1}^{*, \infty}(\bar{\omega})$ compares with $W^{\infty}_{t^{*}-1}(\bar{\omega})$ and repeat the above process. Due to the finiteness of $t^{*}$ we are guaranteed to find a profitable one-step deviation, for otherwise it contradicts the assumption that $\pi^{*}$ is a superior policy to $\pi^g$.  \qed
} 

\begin{remark}
The above lemma is not conditioned on the values of $p_{11}, p_{01}$, and is thus reused in the next section in the case $p_{11} < p_{01}$.
\end{remark}


\subsection{Optimality of myopic sensing}

We begin by introducing a bound on the value function, which is then used in proving the optimality condition. Denote $\delta := p_{11} - p_{01}$ and notice under this section we have $\delta \geq 0$; and we will use $W^\infty(\cdot)$ to denote $W^\infty_t(\cdot), t=1,2,...$ for simplicity.

\begin{lemma}[Boundedness]  
\label{lem:boundedness}
Consider the finite horizon problem (P1) with horizon $T$.  For $1 \leq t \leq T$, $x \geq y$, and $\Delta_t = \overline{\mathcal R}
\cdot \sum_{i = 0}^{T-t} (\beta\cdot \delta)^i$, we have 
\begin{eqnarray}
0 &\leq& W_t(\omega_1,..., \omega_{j-1}, x, \omega_{j+1}, ...,\omega_N)\nonumber \\
&-&W_t(\omega_1,..., \omega_{j-1}, y, \omega_{j+1}, ...,\omega_N) \leq (x-y)\cdot
\Delta_t ~. 
\end{eqnarray}
\end{lemma}

\begin{remark}
A direct consequence of the above result is the following extension to infinite horizon. 
\begin{eqnarray}
&& W^\infty(\omega_1,...,x,...,\omega_N) - W^\infty(\omega_1,...,y,...\omega_N) \nonumber \\
&=& \lim_{T \rightarrow \infty} \{W_1(\omega_1,...,x,...,\omega_N)-W_1(\omega_1,...,y,...,\omega_N)\} \nonumber \\
&\leq&  \lim_{T \rightarrow \infty}\overline{\mathcal R}
\cdot \sum_{i = 0}^{T-1} (\beta\cdot \delta)^i 
=  \frac{(x-y)\cdot \overline{\mathcal R}}{1-\beta\cdot \delta} = (x-y) \Delta_\infty ~. 
\label{eqn:infinite_positive_delta}
\end{eqnarray}
\end{remark}

\begin{lemma}
When $\delta$ satisfies the following condition 
\begin{align}
\frac{\delta}{1-\delta} <   \underline{\mathcal R}/\overline{\mathcal R}, \label{eqn:c1}
\end{align}
there is no profitable one-step deviation for \textbf{(P2)} for any $0 < \beta < 1$. \label{lem:inf1}
\end{lemma}
\com{Suggest we be consistent and use $\delta$ throughout the proofs to replace $p_{11}-p_{01}$ or $p_{01}-p_{11}$ in the next section...} 
\com{Also, not sure where the $p_{11}$ in the numerator comes from... see proof.}

\rev{
The above result appears to suggest that the closer the two values $p_{11}$ and $p_{01}$, the easier it is for the greedy policy to be optimal (though the two quantities $\overline{\cal R}$ and $\underline{\cal R}$ are also functions of $p_{11}$ and $p_{01}$).  The reason is that for a non-greedy policy to outperform the greedy policy, the former must have higher future discounted reward as the latter by definition has higher immediate reward. This, however, is made more difficult when $\delta$ is small, as it has the effect of damping the difference between the two policies.  To illustrate, consider two information states differing in only one element, $x$ vs. $y$. The difference in the immediate reward is a function of $x-y$; however, when propagated to the next time step, the corresponding elements in the information states become $\tau(x)$ and $\tau(y)$, and the difference in the corresponding value functions is now a function of $\tau(x)-\tau(y) = \delta(x-y)$. Thus if $\delta$ is sufficiently small, the difference in future reward will be limited, guaranteeing the optimality of the greedy policy. The details are shown in the proof given in the appendix.  
}
\begin{theorem}
Myopic sensing is optimal for \textbf{(P2)} and \textbf{(P3)} under condition (\ref{eqn:c1}).
\end{theorem}
\emph{Proof.}  Lemma \ref{lem:inf1} combined with Lemma \ref{lem:one-step} immediately imply that myopic sensing is optimal for \textbf{(P2)}.  Since this result holds for any choice of $0< \beta < 1$, the optimality is also true for \textbf{(P3)}. 
\qed

\subsection{A numerical study}
We next show some numerical results to give a sense of the range of $(p_{11},p_{01})$ pairs, $p_{11} \geq p_{01}$, that would guarantee the optimality of myopic sensing. These results are for the case of $(k,m) = (2,1)$, i.e., while sensing 2 channels we only use 1 for transmission.
\begin{figure}[!h]
    \centering
        \includegraphics[width=0.5\textwidth]{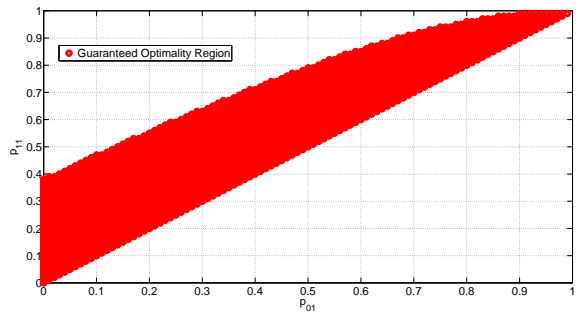}
    \centering
    	\caption{Guaranteed optimality region : case with $p_{11} \geq p_{01}$ } \label{positive}
\end{figure}
From Fig.\ref{positive} we can see when $p_{11}$ is small ( $< 0.5$), almost all pairs of $(p_{01},p_{11})$ would satisfy the optimality condition. As $p_{01}$ increases, the choice of $p_{11}$ becomes more limited. 


\section{Infinite Horizon : $p_{11} < p_{01}$} \label{sec:infinite_negative}

In this section we analyze the infinite horizon problems with negatively correlated channels, i.e., with parameters $p_{11}<p_{01}$. The basic idea is same as in the case of $p_{11} \geq p_{01}$, but the technical details differ; as we show later the difficulties arise mainly from the loss of monotonicity of the value functions with negatively correlated channels.

We start similarly with a lemma regarding the boundedness of the value functions. 
 \begin{lemma} Consider the finite horizon problem (P1) with horizon $T$, and $\forall 1 \leq j \leq N, 1\leq t \leq T$.  Denoting $\delta := p_{01}-p_{11}$, we have 
\begin{align}
(x-y)\cdot \underline{\Delta}_t &\leq W_t(\omega_1,..., \omega_{j-1}, x, \omega_{j+1}, ...,\omega_N)\nonumber \\
 &-W_t(\omega_1,..., \omega_{j-1}, y, \omega_{j+1}, ...,\omega_N) \leq (x-y)\cdot \overline{\Delta}_t
\end{align}
where $\underline{\Delta}_t,\overline{\Delta}_t$ are defined as
\begin{equation}
\underline{\Delta}_t = \left\{
    \begin{array}{cr}
      \frac{1-(\beta\cdot\delta)^{T-t+3}}{1-(\beta\cdot\delta)^{2}}\cdot \eta, & \eta < 0 \\
     0 , & \eta \geq 0.
    \end{array} \right.
\end{equation}
\begin{equation}
\overline{\Delta}_t = \left\{
    \begin{array}{cr}
     \overline{\mathcal R} - \frac{1-(\beta\cdot\delta)^{T-t+3}}{1-(\beta\cdot\delta)^{2}}\cdot \eta, & \eta < 0 \\
     \overline{\mathcal R} , & \eta \geq 0.
    \end{array} \right.
\end{equation}
Here $\eta :=  \underline{\mathcal R} - \beta \cdot (p_{10}-p_{11}) \cdot \overline{\mathcal R}$.\label{lem:bound2}
\end{lemma} 
\begin{remark}
For $\underline{\Delta}_1,\overline{\Delta}_1$ when $T$ goes to infinity we have 
\com{why the superscript ``u''?} \com{was referring to ``uniform''. removed now. }
\begin{align}
\underline{\Delta}^{\infty}_1 &= \min\{\frac{ \eta}{1-(\beta\cdot \delta)^{2}},0\}\\
\overline{\Delta}^{\infty}_1 &= \max\{\overline{\mathcal R}- (\beta\cdot \delta) \cdot \frac{ \eta}{1-(\beta\cdot \delta)^{2}},\overline{\mathcal R} \}~. 
\end{align}
\end{remark}

We next establish the optimality condition for the case $p_{11} < p_{01}$. The argument is similar to that used for the case $p_{11} \geq p_{01}$, i.e., we bound the difference between {immediate rewards} and {future rewards} respectively and compare.  
The detailed proof of this lemma is thus omitted for brevity. 
\begin{lemma}\label{lemma_inf_neg}
Denote by $\delta = p_{01}-p_{11}$.  When the pair $(p_{11}, p_{01})$ satisfies the following condition 
\begin{align}
 \min \{\delta, \frac{1}{2(1-\delta)}\} \leq \underline{\mathcal R}/\overline{\mathcal R}~, 
\end{align}
then there is no profitable one-step deviation for \textbf{(P2)} for any $0 < \beta < 1$.
\end{lemma}
\begin{theorem}
Myopic sensing is optimal for \textbf{(P2)} and \textbf{(P3)} when the condition in Lemma \ref{lemma_inf_neg} is satisfied.
\end{theorem}
\emph{Proof.} The proof follows immediately from the one-step deviation principle. \qed

\subsection{A numerical study}
Here we show similar numerical results on the range of $(p_{11},p_{01})$ pairs, $p_{11} < p_{01}$,  that would guarantee the optimality of myopic sensing according to Lemma \ref{lemma_inf_neg}.  Again we use the case of $(k,m) = (2,1)$.
\begin{figure}[!h]
    \centering
        \includegraphics[width=0.5\textwidth]{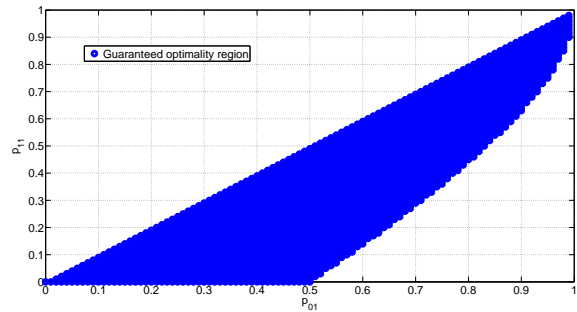}
    \centering
    	\caption{Guaranteed optimality region : case with $p_{11} < p_{01}$ } \label{negative}
\end{figure} 
%
\rev{
This picture appears to be a mirror image (w.r.t. the diagonal $p_{11}=p_{01}$) of the earlier one. 
When $p_{01}$ is small ($< 0.5$), most pairs of $(p_{01},p_{11})$ satisfy our optimality condition. When $p_{01}$ increases, the choice of $p_{11}$ becomes more limited. } 

\section{Discussion} \label{sec:discussion}

In deriving the set of sufficient conditions we have used two different methods: an induction and sample path based argument for the finite horizon problem and a set of bounds for the infinite horizon problems.  In addition, the first set of conditions is on $\beta$, while the second set on $p_{11}$ and $p_{01}$.  
The induction based argument for the finite horizon problem cannot be extended to address the infinite horizon problems; however, the bounding techniques combined with the one-step deviation principle can be applied to obtain alternate sufficient conditions for the finite horizon problem. 
The detail is omitted as the essence of the method remains the same as we have shown in the infinite horizon problems.

\section{Conclusion}\label{sec:conclusion}

This paper we considered a widely studied stochastic control problem arising from opportunistic spectrum access in a multi-channel system, 
where a single wireless transceiver/user with access to $N$ channels, each modeled as an iid discrete-time two-state Markov chain.  In each time step the user is allowed to sense $k\leq N$ channels, and subsequently use up to $m\leq k$ channels out of those sensed to be available. 
%
This problem has previously been studied in various special cases including $m=k=1$ and $m=k\leq N$; it is often cast as a restless bandit problem, with optimality results derived for a myopic policy that seeks to maximize the immediate one-step reward when the two-state Markov chain model is positively correlated.  
%
We derived sufficient conditions under which the myopic policy is optimal for the finite and infinite horizon reward criteria, respectively.  It is shown that these results reduce to those derived in prior studies under the corresponding special cases, and thus may be viewed as a set of unifying optimality conditions. 




\bibliography{osa-IT}

\begin{appendix}
\rev{
\subsection{Proof of Lemma \ref{lemma:mono}} 

We prove this by induction on $t$. Denote by $\bar{\omega}^{+} \subset \bar{\omega}^{'}$ the subset of components that are strictly larger in $\bar{\omega^{'}}$ than in $\bar{\omega}$, i.e., $\bar{\omega}^+=\{\omega_i^{'}, i=1, \cdots, N, ~ \mbox{s. t. } \omega^{'}_i > \omega_i\}$.  

{\em Induction basis:} When $t = T$, the lemma holds due to the increasing property of the one-step expected reward given in Proposition \ref{eqn:prop_increase}. 

{\em Induction step:} Assume the lemma holds for $t+1, \cdots, T$, and consider time $t$.
%
There are two cases: 

\textbf{Case 1.} $\bar{\omega}^{+} \cap \bar{\omega}(\pi^g) = \emptyset$.
In this case since the elements strictly larger in $\omega^{'}$ are not used, the expected one-step rewards under $\bar{\omega}^{'}$ and under $\bar{\omega}$ are the same.  The future reward under $\bar{\omega}^{'}$ is no smaller than that under $\bar{\omega}$ due to the induction hypothesis and the monotonicity of $\tau(\cdot)$, i.e., $\tau(\omega^{'}_j) > \tau(\omega_j)$ for $\omega^{'}_j > \omega_j$.   

\textbf{Case 2.} $\bar{\omega}^{+} \cap \bar{\omega}(\pi^g) \neq \emptyset$.
Consider some $j \in \bar{\omega}^{+} \cap  \bar{\omega}(\pi^g)$, and the state vector $(\omega_j^{'}, \omega_{-j})$;  
it differs from $\omega$ by only one element $\omega^{'}_j$. Using the alternate expression given in (\ref{eqn:W_alt}) we have 
\begin{eqnarray}
&& W_t(\omega^{'}_j, \omega_{-j}) \nonumber \\
&=&  \mathbb E [R_{\pi^g}(\omega^{'}_j, \omega_{-j})] + \nonumber \\
&& \omega^{'}_j \beta \cdot \sum_{\bar{l}_{-j} \in \{0,1\}^{k-1}} q(\bar{l}_{-j}; \omega_{-j}) 
\underbrace{W_{t+1}(p_{11}[\sum_{i\neq j} l_i+1],\tau({\omega_{k+1}}),.., \tau(\omega_{N}), p_{01}[k-\sum_{i\neq j} l_i-1])}_{\rm R_1}+  \nonumber\\
&& (1-\omega^{'}_j) \beta \cdot \sum_{\bar{l}_{-j} \in \{0,1\}^{k-1}} q(\bar{l}_{-j}; \omega_{-j})  
\underbrace{W_{t+1}(p_{11}[\sum_{i\neq j} l_i],\tau({\omega_{k+1}}),.., \tau(\omega_{N}), p_{01}[k-\sum_{i\neq j} l_i])}_{\rm R_2} \nonumber \\
%
&\geq & \mathbb E [R_{\pi^g}(\omega_j, \omega_{-j})] + \nonumber \\
&& \omega_j \beta \cdot \sum_{\bar{l}_{-j} \in \{0,1\}^{k-1}} q(\bar{l}_{-j}; \omega_{-j}) 
\underbrace{W_{t+1}(p_{11}[\sum_{i\neq j} l_i+1],\tau({\omega_{k+1}}),.., \tau(\omega_{N}), p_{01}[k-\sum_{i\neq j} l_i-1])}_{\rm R_1}+  \nonumber\\
&& (1-\omega_j) \beta \cdot \sum_{\bar{l}_{-j} \in \{0,1\}^{k-1}} q(\bar{l}_{-j}; \omega_{-j})  
\underbrace{W_{t+1}(p_{11}[\sum_{i\neq j} l_i],\tau({\omega_{k+1}}),.., \tau(\omega_{N}), p_{01}[k-\sum_{i\neq j} l_i])}_{\rm R_2} \nonumber \\
%
&=&  W_t(\bar\omega)~, 
\end{eqnarray}
where the inequality holds because 
(1) $\mathbb E [R_{\pi^g}(\omega^{'}_j, \omega_{-j})] \geq \mathbb E [R_{\pi^g}(\omega_j, \omega_{-j})]$ by Proposition \ref{eqn:prop_increase}, and 
(2) $\omega^{'}_j \cdot R_1 + (1-\omega^{'}_j) \cdot R_2\geq \omega_j \cdot R_1 + (1-\omega_j) \cdot R_2$ since $\omega^{'}_j>\omega_j$ and $R_1\geq R_2$ due to the induction hypothesis. 
%
We can now repeat the above process by introducing another element $k\in \bar{\omega}^{+}\cap\bar{\omega}(\pi^g), k\neq j$, and obtain similarly, $W_t(\omega^{'}_k, \omega^{'}_j, \omega_{-j,-k}) \geq W_t(\omega^{'}_j, \omega_{-j}) \geq W_t(\bar{\omega})$. When all elements in $\bar{\omega}^{+}\cap\bar{\omega}(\pi^g), k\neq j$ have been exhausted we obtain $W_t(\bar{\omega}^{'}) \geq W_t(\bar{\omega})$. 
The induction steps is thus completed. 
}

\rev{
\subsection{Proof of Lemma \ref{lem:key}}

The two inequalities L1 and L2 will be shown together using an induction on $t$.  

{\em Induction basis:} For $t=T$, L1 holds because in this case 
\begin{eqnarray}
&& W_T(\omega_N,\omega_1,...,\omega_{N-1}) - W_T(\omega_1,...,\omega_{N}) \nonumber \\
&=& \mathbb E[R_{\pi^g}(\omega_N,\omega_1,...,\omega_{N-1})] - E[R_{\pi^g}(\omega_1,...,\omega_{N})] 
\nonumber \\
&\leq&   E[R_{\pi^g}(\omega_N = p_{11},\omega_1,...,\omega_{N-1})] - E[R_{\pi^g}(\omega_1,...,\omega_k = p_{01},...,\omega_{N})] \nonumber \\
&\leq&   E[R_{\pi^g}(\omega_N = p_{11},\omega_1,...,\omega_{N-1})] - E[R_{\pi^g}(\omega_k = p_{01}, \omega_1,...,\omega_{N})]  \leq \overline{\cal R} ~, 
\end{eqnarray}
using the increasing property, Proposition \ref{eqn:prop_increase}, of the expected one-step reward.
L2 holds at $T$ due to the same reason.  Assume both L1 and L2 hold for times $t+1, \cdots, T$. 

{\em Induction step:} 
We will employ a sample-path argument by calculating the quantities on the LHS (RHS) of these two inequalities conditioned on the outcome of sensing specific channels.  
Consider first L1.  At time $t$, the LHS selects channels $\{N, 1, \cdots, k-1\}$ while the RHS selects channels $\{1, \cdots, k\}$.  Thus the two sides differ only in channels $\{k, N\}$. For simplicity we denote by $\text{LHS}|_{i,j}$ (resp. $\text{RHS}|_{i,j}$) the value of the LHS  (resp. RHS) of L1 conditioned on the realizations of channels $k$ and $N$ being $i$ and $j$, respectively, where $i,j \in \{0,1\}$.  
Denote by $\pi^g_{ k-1}:=\{1, 2, \cdots, k-1\}$; this is the common set of channels sensed by both sides.  
Also recall the notation $\bar{l}_{-k} = \{l_1, \cdots, l_{k-1}\}$. 

\textbf{Case 1.} $(k, N) = (``1'', ``0'')$: channel $k$ has state realization ``1'' and channel $N$ ``0''. 
In this case we have
\begin{eqnarray}
\text{LHS}|_{1,0} &=& \overline{\mathcal R} + \mathbb E [R_{\pi^g}(0,\omega_1,...,\omega_{N-1})] + \beta \cdot \sum_{\bar{l}_{-k} \in \{0,1\}^{k-1}} q(\bar{l}_{-k}; \omega_{-k})\cdot \nonumber \\
&&W_{t+1}(p_{11}[\sum_{i=1}^{k-1} l_i],\tau(\omega_k) = p_{11}, \cdots, \tau(\omega_{N-1}), p_{01}[k-\sum_{i=1}^{k-1} l_i]) \nonumber \\
\text{RHS}|_{1,0} &=& \mathbb E [R_{\pi^g}(1,\omega_2,...,\omega_N)] + \beta \cdot \sum_{\bar{l}_{-k} \in \{0,1\}^{k-1}} q(\bar{l}_{-k}; \omega_{-k}) \cdot \nonumber \\
&&W_{t+1}(p_{11}[\sum_{i=1}^{k-1} l_i], \tau(\omega_k) = p_{11}, \cdots, \tau(\omega_{N-1}), \tau(\omega_N)= p_{01}, p_{01}[k-1-\sum_{i=1}^{k-1} l_i]) ~. 
\end{eqnarray}
By the definition of $\overline{\mathcal R}$ we have $\overline{\mathcal R} + \mathbb E [R_{\pi^g}(0,\omega_1,...,\omega_{N-1})] - \mathbb E [R_{\pi^g}(1,\omega_2,...,\omega_N)] \geq 0$, thus $\text{LHS}|_{1,0} \geq \text{RHS}|_{1,0}$.  

\textbf{Case 2.} $(k, N) = (``1'', ``1'')$: both channels $k$ and $N$ have state realizations ``1''.  In this case 
\begin{eqnarray}
\text{LHS}|_{1,1} &=& \overline{\mathcal R} + \mathbb E [R_{\pi^g}(1,\omega_1,...,\omega_{N-1})]  + \beta \cdot \sum_{\bar{l}_{-k} \in \{0,1\}^{k-1}} q(\bar{l}_{-k}; \omega_{-k})\cdot \nonumber \\
&&W_{t+1}(p_{11}[\sum_{i=1}^{k-1} l_i+1],\tau(\omega_k) = p_{11}, \cdots, \tau(\omega_{N-1}), p_{01}[k-1-\sum_{i=1}^{k-1} l_i]) \nonumber \\
\text{RHS}|_{1,1} &=& \mathbb E [R_{\pi^g}(1,\omega_2,...,\omega_N)] + \beta \cdot \sum_{\bar{l}_{-k} \in \{0,1\}^{k-1}} q(\bar{l}_{-k}; \omega_{-k}) \cdot \nonumber \\
&&W_{t+1}(p_{11}[\sum_{i=1}^{k-1} l_i+1], \tau(\omega_{k+1}), \cdots, \tau(\omega_N)= p_{11}, p_{01}[k-1-\sum_{i=1}^{k-1} l_i]) ~. 
\end{eqnarray}
$\text{LHS}|_{1,0} \geq \text{RHS}|_{1,0}$ because (1) $\overline{\mathcal R}\geq 0$, (2) $\mathbb E [R_{\pi^g}(1,\omega_1,...,\omega_{N-1})] = \mathbb E [R_{\pi^g}(1,\omega_2,...,\omega_N)]$, and (3) by repeatedly using the induction hypothesis of L2 (successively moving $\tau(\omega_k)=p_{11}$ to the right or down the ordered list). 

\textbf{Case 3.} $(k, N)=(``0'', ``0'')$: both channels $k$ and $N$ have state realizations ``0''.  We have
\begin{eqnarray}
\text{LHS}|_{0,0} &=& \overline{\mathcal R} + \mathbb E [R_{\pi^g}(0,\omega_1,...,\omega_{N-1})] + \beta \cdot \sum_{\bar{l}_{-k} \in \{0,1\}^{k-1}} q(\bar{l}_{-k}; \omega_{-k})\cdot \nonumber \\
&&W_{t+1}(p_{11}[\sum_{i=1}^{k-1} l_i], \tau(\omega_k) = p_{01}, \cdots, \tau(\omega_{N-1}), p_{01}[k-\sum_{i=1}^{k-1} l_i]) \nonumber \\
\text{RHS}|_{0,0} &=& \mathbb E [R_{\pi^g}(0,\omega_1,...,\omega_N)] + \beta \cdot \sum_{\bar{l}_{-k} \in \{0,1\}^{k-1}} q(\bar{l}_{-k}; \omega_{-k}) \cdot \nonumber \\
&&W_{t+1}(p_{11}[\sum_{i=1}^{k-1} l_i], \tau(\omega_{k+1}), \cdots, \tau(\omega_{N-1}), \tau(\omega_N)= p_{01}, p_{01}[k-\sum_{i=1}^{k-1} l_i]) ~. 
\end{eqnarray}
%
Using the induction hypothesis of both L1 and L2 we have 
\begin{eqnarray}
\text{LHS}|_{0,0} &\geq& \mathbb E [R_{\pi^g}(0,\omega_1,...,\omega_{N-1})] + \beta \cdot \sum_{\bar{l}_{-k} \in \{0,1\}^{k-1}} q(\bar{l}_{-k}; \omega_{-k})\cdot \nonumber \\
&&(\overline{\mathcal R} + W_{t+1}(p_{11}[\sum_{i=1}^{k-1} l_i], \tau(\omega_k) = p_{01}, \cdots, \tau(\omega_{N-1}), p_{01}[k-\sum_{i=1}^{k-1} l_i]))\nonumber \\
&\geq& \mathbb E [R_{\pi^g}(0,\omega_1,...,\omega_N)] + \beta \cdot \sum_{\bar{l}_{-k} \in \{0,1\}^{k-1}} q(\bar{l}_{-k}; \omega_{-k})\cdot \nonumber \\
&&(\overline{\mathcal R} + W_{t+1}(\tau(\omega_{k})=p_{01}, p_{11}[\sum_{i=1}^{k-1} l_i], \tau(\omega_{k+1}), \cdots, \tau(\omega_{N-1}), p_{01}[k-\sum_{i=1}^{k-1} l_i])) \nonumber \\
&\geq& \mathbb E [R_{\pi^g}(0,\omega_1,...,\omega_N)] + \beta \cdot \sum_{\bar{l}_{-k} \in \{0,1\}^{k-1}} q(\bar{l}_{-k}; \omega_{-k})\cdot \nonumber \\
&&W_{t+1}(p_{11}[\sum_{i=1}^{k-1} l_i], \tau(\omega_{k+1}), \cdots, \tau(\omega_{N-1}), p_{01}[k-\sum_{i=1}^{k-1} l_i], \tau(\omega_k)=p_{01}) \nonumber \\
&=& \text{LHS}|_{0,0} ~, 
\end{eqnarray}
where the first inequality is due to the fact that $q(\cdot)$ forms a probability distribution and $\beta \overline{\mathcal R} < \overline{\mathcal R}$, the second due to the induction hypothesis of L2, and the third due to the induction hypothesis of L1. 

\textbf{Case 4.} $(k, N)=(``0'', ``1'')$: channels $k$ and $N$ have state realizations ``0'' and ``1'', respectively.  We have
\begin{eqnarray}
\text{LHS}|_{0,1} &=& \overline{\mathcal R} + \mathbb E [R_{\pi^g}(1,\omega_1,...,\omega_{N-1})] + \beta \cdot \sum_{\bar{l}_{-k} \in \{0,1\}^{k-1}} q(\bar{l}_{-k}; \omega_{-k})\cdot \nonumber \\
&&W_{t+1}(p_{11}[\sum_{i=1}^{k-1} l_i+1], \tau(\omega_k) = p_{01}, \cdots, \tau(\omega_{N-1}), p_{01}[k-1-\sum_{i=1}^{k-1} l_i]) \nonumber \\
\text{RHS}|_{0,1} &=& \mathbb E [R_{\pi^g}(0,\omega_1,...,\omega_N)] + \beta \cdot \sum_{\bar{l}_{-k} \in \{0,1\}^{k-1}} q(\bar{l}_{-k}; \omega_{-k}) \cdot \nonumber \\
&&W_{t+1}(p_{11}[\sum_{i=1}^{k-1} l_i], \tau(\omega_{k+1}), \cdots, \tau(\omega_{N-1}), \tau(\omega_N)= p_{11}, p_{01}[k-\sum_{i=1}^{k-1} l_i]) ~. 
\end{eqnarray}
%
\begin{eqnarray}
\text{LHS}|_{0,1} &\geq& \mathbb E [R_{\pi^g}(1,\omega_1,...,\omega_{N-1})] + \beta \cdot \sum_{\bar{l}_{-k} \in \{0,1\}^{k-1}} q(\bar{l}_{-k}; \omega_{-k})\cdot \nonumber \\
&&(\overline{\mathcal R} + W_{t+1}(p_{11}[\sum_{i=1}^{k-1} l_i+1], \tau(\omega_k) = p_{01}, \cdots, \tau(\omega_{N-1}), p_{01}[k-1-\sum_{i=1}^{k-1} l_i]))\nonumber \\
&\geq& \mathbb E [R_{\pi^g}(0,\omega_1,...,\omega_N)] + \beta \cdot \sum_{\bar{l}_{-k} \in \{0,1\}^k} q(\bar{l}_{-k}; \omega_{-k})\cdot \nonumber \\
&&W_{t+1}(p_{11}[\sum_{i=1}^{k-1} l_i+1], \tau(\omega_{k+1}), \cdots, \tau(\omega_{N-1}), p_{01}[k-1-\sum_{i=1}^{k-1} l_i], \tau(\omega_k)=p_{01}) \nonumber \\
&\geq& \mathbb E [R_{\pi^g}(0,\omega_1,...,\omega_N)] + \beta \cdot \sum_{\bar{l}_{-k} \in \{0,1\}^k} q(\bar{l}_{-k}; \omega_{-k})\cdot \nonumber \\
&&W_{t+1}(p_{11}[\sum_{i=1}^{k-1} l_i], \tau(\omega_{k+1}), \cdots, \tau(\omega_{N-1}), p_{11}, p_{01}[k-\sum_{i=1}^{k-1} l_i]) \nonumber \\
&=& \text{LHS}|_{0,0} ~, 
\end{eqnarray}
where the first inequality is due to Proposition \ref{eqn:prop_increase}, the second due to induction hypothesis of L2 (moving $\tau(\omega_k)=p_{01}$ to the front/left of the list, following by induction hypothesis of L1 (moving $\tau(\omega_k)=p_{01}$ to the end/right of the list), and the third due to the induction hypothesis of L2. 

We have now established the induction step of L1, thus proving L1.  Next we consider L2 at time $t$.  In the case when $j \leq k-1$, both $x$ and $y$ are used by both sides, so $\text{LHS} = \text{RHS}$.  In the case when $j\geq k+1$, neither channel $j$ nor $j+1$ is used.  Thus both sides will return the same one-step reward.  The difference between $x$ and $y$ propagates to the future reward term $W_{t+1}(\cdot)$.  However, due to the fact that $\tau(x)\geq \tau(y)$, using the induction hypothesis of L2 we conclude $\text{LHS}\geq \text{RHS}$.   
%
%
%
%
%

It remains to check the case $j=k$.  In this case we single out both $x$ and $y$: 
\begin{eqnarray*}
&& \text{LHS} =\mathbb E [R_{\pi^g}(x, \omega_1,...,\omega_{k+1},...,\omega_N)] \nonumber \\
&& + \beta \{ x\cdot y \underbrace{\sum_{\bar{l}_{-k} \in \{0,1\}^{k-1}} q(\bar{l}_{-k}; \omega_{-k})\cdot 
W_{t+1}(p_{11}[\sum_{k=1}^{k-1} l_i+1], p_{11}, \tau(\omega_{k+2}), \cdots, \tau(\omega_N), p_{01}[k-1-\sum_{i=1}^{k-1} l_i])}_{\rm R1}  \nonumber \\
&& + (1-x)\cdot y \underbrace{\sum_{\bar{l}_{-k} \in \{0,1\}^{k-1}} q(\bar{l}_{-k}; \omega_{-k})\cdot 
W_{t+1}(p_{11}[\sum_{k=1}^{k-1} l_i], p_{11}, \tau(\omega_{k+2}), \cdots, \tau(\omega_N), p_{01}[k-\sum_{i=1}^{k-1} l_i])}_{\rm R2} \nonumber \\
&& + x\cdot (1-y) \underbrace{\sum_{\bar{l}_{-k} \in \{0,1\}^{k-1}} q(\bar{l}_{-k}; \omega_{-k})\cdot 
W_{t+1}(p_{11}[\sum_{k=1}^{k-1} l_i+1], p_{01}, \tau(\omega_{k+2}), \cdots, \tau(\omega_N), p_{01}[k-1-\sum_{i=1}^{k-1} l_i])}_{\rm R3} \nonumber \\
&& + (1-x)\cdot (1-y) \underbrace{\sum_{\bar{l}_{-k} \in \{0,1\}^{k-1}} q(\bar{l}_{-k}; \omega_{-k})\cdot 
W_{t+1}(p_{11}[\sum_{k=1}^{k-1} l_i], p_{01}, \tau(\omega_{k+2}), \cdots, \tau(\omega_N), p_{01}[k-\sum_{i=1}^{k-1} l_i])}_{\rm R4} \}  \label{12l}
\end{eqnarray*}
Similarly, 
\begin{eqnarray*}
&& \text{RHS} =\mathbb E [R_{\pi^g}(y, \omega_1,...,\omega_{k},...,\omega_N)] \nonumber \\
&& + \beta \{ x\cdot y \underbrace{\sum_{\bar{l}_{-k} \in \{0,1\}^{k-1}} q(\bar{l}_{-k}; \omega_{-k})\cdot 
W_{t+1}(p_{11}[\sum_{k=1}^{k-1} l_i+1], p_{11}, \tau(\omega_{k+2}), \cdots, \tau(\omega_N), p_{01}[k-1-\sum_{i=1}^{k-1} l_i])}_{\rm R1}  \nonumber \\
&& + (1-x)\cdot y \underbrace{\sum_{\bar{l}_{-k} \in \{0,1\}^{k-1}} q(\bar{l}_{-k}; \omega_{-k})\cdot 
W_{t+1}(p_{11}[\sum_{k=1}^{k-1} l_i+1], p_{01}, \tau(\omega_{k+2}), \cdots, \tau(\omega_N), p_{01}[k-1-\sum_{i=1}^{k-1} l_i])}_{\rm R3} \nonumber \\
&& + x\cdot (1-y) \underbrace{\sum_{\bar{l}_{-k} \in \{0,1\}^{k-1}} q(\bar{l}_{-k}; \omega_{-k})\cdot 
W_{t+1}(p_{11}[\sum_{k=1}^{k-1} l_i], p_{11}, \tau(\omega_{k+2}), \cdots, \tau(\omega_N), p_{01}[k-\sum_{i=1}^{k-1} l_i])}_{\rm R2} \nonumber \\
&& + (1-x)\cdot (1-y) \underbrace{\sum_{\bar{l}_{-k} \in \{0,1\}^{k-1}} q(\bar{l}_{-k}; \omega_{-k})\cdot 
W_{t+1}(p_{11}[\sum_{k=1}^{k-1} l_i], p_{01}, \tau(\omega_{k+2}), \cdots, \tau(\omega_N), p_{01}[k-\sum_{i=1}^{k-1} l_i])}_{\rm R4} \}  \label{12l}
\end{eqnarray*}

Thus we have 
\begin{eqnarray}
&& \text{LHS} - \text{RHS} = \mathbb E[R_{\pi^g}(x, \omega_{-k})] -\mathbb E[R_{\pi^g}(y, \omega_{-k})] + \beta (x-y) (R3-R2) \nonumber \\
&= & (x-y) (\mathbb E[R_{\pi^g}(1, \omega_{-k})] -\mathbb E[R_{\pi^g}(0, \omega_{-k})]) + 
\beta (x-y) \sum_{\bar{l}_{-k} \in \{0,1\}^{k-1}} q(\bar{l}_{-k}; \omega_{-k})\cdot  \nonumber \\
&& \left(W_{t+1}(p_{11}[\sum_{k=1}^{k-1} l_i+1], p_{01}, \tau(\omega_{k+2}), \cdots, \tau(\omega_N), p_{01}[k-1-\sum_{i=1}^{k-1} l_i]) \right.\nonumber \\
&& ~~~ \left. -
W_{t+1}(p_{11}[\sum_{k=1}^{k-1} l_i], p_{11}, \tau(\omega_{k+2}), \cdots, \tau(\omega_N), p_{01}[k-\sum_{i=1}^{k-1} l_i]) \right) \nonumber \\
&\geq & (x-y)\underline{\mathcal R} + \beta (x-y) \sum_{\bar{l}_{-k} \in \{0,1\}^{k-1}} q(\bar{l}_{-k}; \omega_{-k})\cdot \nonumber \\
&& \left(W_{t+1}(p_{01}, p_{11}[\sum_{k=1}^{k-1} l_i+1], \tau(\omega_{k+2}), \cdots, \tau(\omega_N), p_{01}[k-1-\sum_{i=1}^{k-1} l_i]) \right. \nonumber \\
&& ~~~ \left. -
W_{t+1}(p_{11}[\sum_{k=1}^{k-1} l_i], p_{11}, \tau(\omega_{k+2}), \cdots, \tau(\omega_N), p_{01}[k-\sum_{i=1}^{k-1} l_i]) \right) \nonumber \\
&\geq & (x-y) \underline{\mathcal R} - \beta(x-y) \overline{\mathcal R} ~, 
\end{eqnarray}
where the first inequality is due to the definition of $\underline{\mathcal R}$ and the use of the induction hypothesis of L2, and the second inequality due to the induction hypothesis of L1. 
Therefore if 
$\beta \leq   \underline{\mathcal R}/\overline{\mathcal R} \label{12}$, then we will have LHS $\geq$ RHS, completing the induction step of L2. 
} 

\subsection{Proof of Lemma \ref{lem:key2}}

The three inequalities L3, L4 and L5 are shown together using an induction on $t$. 

{\em Induction basis:} At time $T$, L3 becomes 
$\gamma + \mathbb E[R_{\pi^g}(\omega_1, \cdots, \omega_{N}, \omega_{1})] \geq  \mathbb E[R_{\pi^g}(\omega_1, \cdots, \omega_N)]$. 
This holds because 
\begin{eqnarray}
&& \mathbb E[R_{\pi^g}(\omega_1, \cdots, \omega_{N})] -  \mathbb E[R_{\pi^g}(\omega_2, \cdots, \omega_N, \omega_1)] \nonumber \\
&\leq& \mathbb E[R_{\pi^g}(\omega_1=1, \cdots, \omega_{N})]-\mathbb E[R_{\pi^g}(\omega_2, \cdots, \omega_{k+1}=0,\omega_N, \omega_1)] \nonumber \\
&\leq& \overline{\mathcal R} \leq \frac{\overline{\mathcal R}}{1-\beta} = \gamma. 
\end{eqnarray}
Similarly, L4 holds at time $T$ because 
\begin{eqnarray}
&& \mathbb E[R_{\pi^g}(\omega_1, \cdots, \omega_{N})] -  \mathbb E[R_{\pi^g}(\omega_N, \omega_1, \cdots, \omega_{N-1})] \nonumber \\
&\leq& \mathbb E[R_{\pi^g}(\omega_1, \cdots, \omega_k=1, \cdots, \omega_{N})]-\mathbb E[R_{\pi^g}(\omega_N=0, \omega_1, \cdots,\omega_{N-1})] \nonumber \\
&\leq& \overline{\mathcal R} \leq \frac{\overline{\mathcal R}}{1-\beta} = \gamma. 
\end{eqnarray}
%
L5 holds at $T$ due to the increasing property (Proposition \ref{eqn:prop_increase}) of the expected one-step reward. Assume L3, L4 and L5 hold for times $t+1, \cdots, T$. 

{\em Induction step:} We will again employ a sample-path argument conditioned on the outcome of sensing specific channels.  
Consider first L3. At time $t$, the LHS selects channels $\{2,3, \cdots, k+1\}$ while the RHS selects channels $\{1, \cdots, k\}$.  Thus the two sides differ only in channels $\{1, k+1\}$. 

\textbf{Case 1.}  $(1, k+1) = (``0'', ``0'')$: both channels $1$ and $k+1$ have state realization ``0''. 
In this case 
\begin{align}
\text{LHS}|_{0,0} &= \gamma + \mathbb E [R_{\pi^g}(0,\omega_2, \cdots, \omega_k, \omega_{k+1}, \cdots, \omega_N, \omega_1)] + \beta\cdot \sum_{\bar{l}_{-1} \in \{0,1\}^{k-1}}q(\bar{l}_{-1}; \omega_{-1})\cdot \nonumber \\
&W_{t+1}(p_{01}[k-\sum_{i=2}^{k} l_i], \tau(\omega_1) = p_{01}, \tau(\omega_N), \cdots, \tau(\omega_{k+2}), p_{11}[\sum_{i=2}^{k} l_i]) \nonumber \\
\text{RHS}|_{0,0} &=  \mathbb E [R_{\pi^g}(0,\omega_{2},...,\omega_{N})] + \beta\cdot \sum_{\bar{l}_{-1} \in \{0,1\}^{k-1}}q(\bar{l}_{-1}; \omega_{-1})\cdot \nonumber \\
&W_{t+1}(p_{01}[k-\sum_{i=2}^{k} l_i],\tau(\omega_N), \cdots, \tau(\omega_{k+1}), \tau(\omega_{k+1}) = p_{01},p_{11}[\sum_{i=2}^{k} l_i])
\end{align}
By the induction hypothesis of L5 we have $\text{LHS} \geq \text{RHS}$.

\textbf{Case 2.} $(1, k+1) = (``1'', ``0'')$: channel $1$ has state realization ``1'' and channel $k+1$ ``0''. 
In this case
\begin{align}
\text{LHS}|_{1,0} &= \gamma +\mathbb E [R_{\pi^g}(0,\omega_2,\omega_3,...,\omega_1)] + \beta\cdot\sum_{\bar{l}_{-1} \in \{0,1\}^{k-1}}q(\bar{l}_{-1}; \omega_{-1}) \cdot \nonumber \\
&W_{t+1}(p_{01}[k-\sum_{i=2}^{k} l_i],\tau(\omega_1) = p_{11}, \tau(\omega_N), ...,\tau(\omega_{k+2}), p_{11}[\sum_{i=2}^{k} l_i]) \nonumber \\
\text{RHS}|_{1,0} &= \mathbb E [R_{\pi^g}(1,\omega_{2},...,\omega_{N})] + \beta\cdot \sum_{\bar{l}_{-1} \in \{0,1\}^{k-1}}q(\bar{l}_{-1}; \omega_{-1})\cdot \nonumber \\
&W_{t+1}(p_{01}[k-1-\sum_{i=2}^{k} l_i],\tau(\omega_N),...,\tau(\omega_{k+1}) = p_{01},p_{11}[\sum_{i=2}^{k} l_i + 1])
\end{align}
Since $\gamma = \frac{1}{1-\beta}\cdot \overline{\mathcal R} = \overline{\mathcal R} + \beta\cdot \gamma$, we have 
\begin{eqnarray}
\text{LHS}|_{1,0} &=& \overline{\mathcal R} +\mathbb E [R_{\pi^g}(0,\omega_2,\omega_3,...,\omega_1)] + \beta\cdot\sum_{\bar{l}_{-1} \in \{0,1\}^{k-1}}q(\bar{l}_{-1}; \omega_{-1}) \cdot \nonumber \\
&& \left(\gamma + W_{t+1}(p_{01}[k-\sum_{i=2}^{k} l_i],\tau(\omega_1) = p_{11}, ...,p_{11}[\sum_{i=2}^{k} l_i]) \right) \nonumber \\
 &\geq& \mathbb E [R_{\pi^g}(1,\omega_{2},...,\omega_{N})] + \beta\cdot\sum_{\bar{l}_{-1} \in \{0,1\}^{k-1}}q(\bar{l}_{-1}; \omega_{-1}) \cdot \nonumber \\
 && \left(\gamma  + W_{t+1}(p_{11}, p_{01}[k-1-\sum_{i=2}^{k} l_i],...,p_{01}, p_{11}[\sum_{i=2}^{k} l_i]) \right) \nonumber \\
 &\geq& \mathbb E [R_{\pi^g}(1,\omega_{2},...,\omega_{N})] + \beta\cdot\sum_{\bar{l}_{-1} \in \{0,1\}^{k-1}}q(\bar{l}_{-1}; \omega_{-1}) \cdot \nonumber \\
&& W_{t+1}(p_{01}[k-1-\sum_{i=2}^{k} l_i],..., p_{01}, p_{11}[\sum_{i=2}^{k} l_i+1])
\nonumber \\
&= & \text{RHS}|_{1,0} , 
\end{eqnarray}
where the first inequality is due to the definition of $\overline{\mathcal R}$ and the use of the induction hypothesis of L5 and the second inequality is due to the induction hypothesis of L4. 

\textbf{Case 3. } $(1, k+1)=(``0'', ``1'')$: channels $1$ and $k+1$ have realizations ``0'' and ``1'', respectively.  We have
\begin{align}
\text{LHS}|_{0,1} &= \gamma + \mathbb E [R_{\pi^g}(1,\omega_2, ..., \omega_k, \omega_{k+1}, ...,\omega_N, \omega_1)] + \beta\cdot \sum_{\bar{l}_{-1} \in \{0,1\}^{k-1}}q(\bar{l}_{-1}; \omega_{-1})\cdot \nonumber \\
&W_{t+1}(p_{01}[k-1-\sum_{i=2}^{k} l_i],\tau(\omega_1) = p_{01}, \tau(\omega_N), ..., \tau(\omega_{k+2}), p_{11}[\sum_{i=2}^{k} l_i+1]) \nonumber \\
\text{RHS}|_{0,1} &=  \mathbb E [R_{\pi^g}(0,\omega_{2},...,\omega_{N})] + \beta\cdot  \sum_{\bar{l}_{-1} \in \{0,1\}^{k-1}}q(\bar{l}_{-1}; \omega_{-1})\cdot \nonumber \\
&W_{t+1}(p_{01}[k-\sum_{i=2}^{k} l_i], \tau(\omega_N), ...,\tau(\omega_{k+1}) = p_{11}, p_{11}[\sum_{i=2}^{k} l_i]) 
\end{align} 
Since the second part of both $\text{LHS}|_{0,1}$ and $\text{RHS}|_{0,1}$ are identical, we have 
$\text{LHS}|_{0,1} \geq \text{RHS}|_{0,1}$ using the definition of $\gamma$ and $\overline{\mathcal R}$.

\textbf{Case 4.} $(1, k+1) = (``1'', ``1'')$: both channels have state realization ``1''.  In this case 
\begin{align}
\text{LHS}|_{1,1} &= \gamma + \mathbb E [R_{\pi^g}(1,\omega_2,..., \omega_k, \omega_{k+1}, ..., \omega_N, \omega_1)] + \beta\cdot  \sum_{\bar{l}_{-1} \in \{0,1\}^{k-1}}q(\bar{l}_{-1}; \omega_{-1})\cdot \nonumber \\
&W_{t+1}(p_{01}[k-1-\sum_{i=2}^{k} l_i],\tau(\omega_1) = p_{11}, \tau(\omega_N), ..., \tau(\omega_{k+1}), p_{11}[\sum_{i=2}^{k} l_i+1]) 
\nonumber \\
\text{RHS}|_{1,1} &= \mathbb E [R_{\pi^g}(1,\omega_{2},...,\omega_{N})] + \beta\cdot \sum_{\bar{l}_{-1} \in \{0,1\}^{k-1}}q(\bar{l}_{-1}; \omega_{-1})\cdot \nonumber \\
&W_{t+1}(p_{01}[k-1-\sum_{i=2}^{k} l_i], \tau(\omega_N), ..., \tau(\omega_{k+1}) = p_{11}, p_{11}[\sum_{i=2}^{k} l_i +1])
\end{align}
Using a similar method as in Case 2, $\text{LHS}|_{1,1} \geq \text{RHS}|_{1,1}$ holds because 
\begin{align}
& \gamma + W_{t+1}(p_{01}[k-1-\sum_{i=2}^{k} l_i],\tau(\omega_1) = p_{11}, \tau(\omega_N), ..., \tau(\omega_{k+2}), p_{11}[\sum_{i=2}^{k} l_i+1])\nonumber \\
 &\geq \gamma + W_{t+1}(\tau(\omega_1) = p_{11},p_{01}[k-1-\sum_{i=2}^{k} l_i], \tau(\omega_N), ..., \tau(\omega_{k+2}), p_{11}[\sum_{i=2}^{k} l_i+1])\nonumber \\
 &\geq W_{t+1}(p_{01}[k-1-\sum_{i=2}^{k} l_i], \tau(\omega_N), ..., \tau(\omega_{k+2}), \omega_{1} = p_{11},p_{11}[\sum_{i=2}^{k} l_i]), 
\end{align}
using the induction hypothesis of L5 and L4, respectively.  
L3 is thus proven. 

L4 can be shown in the same way L3 is proven above, while L5 can be shown in the same way L2 was proven in Lemma \ref{lem:key}; the details are thus omitted. 

\subsection{Proof of Lemma \ref{lem:boundedness}}
The lower bound is trivial as for finite time horizon problem the monotonicity is already proven in Lemma \ref{lem:key}, L2 time uniformly and we know it can be extended to the infinite horizon problem by simply taking the limitation. 
We prove the upper bound by induction on $t$.  Before proceeding, we note that by definition $\Delta_T = \overline{\mathcal R}$ and $\Delta_t = \overline{\mathcal R} (1+\beta(p_{11}-p_{01}) \sum_{i=0}^{T-t-1} \beta^i \cdot (p_{11}-p_{01})^i) = \overline{\mathcal R} + \beta(p_{11}-p_{01}) \Delta_{t+1}$. 

At time $T$ there are two cases.

\textbf{Case 1.} $j > k$. In this case $W_t(\omega_1,..., \omega_{j-1}, x, \omega_{j+1}, ...,\omega_N) = W_t(\omega_1,..., \omega_{j-1}, y, \omega_{j+1}, ...,\omega_N)$ while $\Delta_T \geq 0$, so the inequality holds. 

\textbf{Case 2.} $j \leq k$. In this case we have 
\begin{eqnarray}
W_t(\omega_1,..., \omega_{j-1}, x, \omega_{j+1}, ...,\omega_N) - W_t(\omega_1,..., \omega_{j-1}, y, \omega_{j+1}, ...,\omega_N) \leq (x-y)\cdot \overline{\mathcal R} = (x-y) \cdot \Delta_T. 
\end{eqnarray}

Now suppose the equalities hold for times $t+1,...,T-1$. Consider time $t$.  Again we have two cases.

\textbf{Case 1.} $j > k$. In this case the immediate reward does not differ between the two belief states as under both the set of sensed channels is identical; we thus have
\begin{eqnarray} 
&& W_t(\omega_1,..., \omega_{j-1}, x, \omega_{j+1}, ...,\omega_N) - W_t(\omega_1,..., \omega_{j-1}, y, \omega_{j+1}, ...,\omega_N) \nonumber \\
&=& \beta \cdot \sum_{\bar{l} \in \{0,1\}^k}q(\bar{l}; \omega) \cdot \left(W_{t+1}(p_{11}[\sum_{i=1}^k l_i],...,\tau(x),...,p_{01}[k-\sum_{i=1}^k l_i]) \right. \nonumber \\ 
&& ~~- \left. W_{t+1}(p_{11}[\sum_{i=1}^k l_i],...,\tau(y),...,p_{01}[k-\sum_{i=1}^k l_i])\right) \nonumber \\
&\leq&  \beta \cdot (\tau(x) - \tau(y))\cdot \Delta_{t+1} \nonumber \\
&=&  (x-y)\cdot \beta \cdot (p_{11}-p_{01})\cdot \Delta_{t+1} \nonumber \\
&\leq& (x-y) (\Delta_t - \overline{\mathcal R}) \leq (x-y) \Delta_t ~. 
\end{eqnarray}
where the first inequality is due to the induction hypothesis.

\textbf{Case 2.} $j \leq k$. In this case we have 
\begin{eqnarray}
&& W_t(\omega_1,..., \omega_{j-1}, x, \omega_{j+1}, ...,\omega_N) - W_t(\omega_1,..., \omega_{j-1}, y, \omega_{j+1}, ...,\omega_N) \nonumber \\
& = & \mathbb E [R_{\pi_g}(x,\omega_{-j})] - \mathbb E [R_{\pi_g}(y,\omega_{-j})] \nonumber \\
&& + (x-y) \cdot \beta
\cdot \sum_{\bar{l}_{-j} \in \{0,1\}^{k-1}}q(\bar{l}_{-j}; \omega_{-j})\cdot W_{t+1}(p_{11}[\sum_{i=1, i\neq j}^{k} l_i],p_{11},...,p_{01}[k-1-\sum_{i=1, i\neq j}^{k} l_i])\nonumber \\
&& +((1-x)-(1-y)) \cdot \beta \cdot \sum_{\bar{l}_{-j} \in \{0,1\}^{k-1}}q(\bar{l}_{-j}; \omega_{-j})\cdot W_{t+1}(p_{11}[\sum_{i=1, i\neq j}^{k} l_i],...,p_{01},p_{01}[k-1-\sum_{i=1, i\neq j}^{k}
l_i]) \nonumber \\
&\leq& (x-y)\overline{\mathcal R} + (x-y)\cdot \beta \cdot (p_{11}-p_{01})\cdot \Delta_{t+1} \nonumber\\
&=& (x-y)\cdot \Delta_t ~. 
\end{eqnarray}
This completes the induction step, thus proving the lemma.

\subsection{Proof of Lemma \ref{lem:inf1}}

For a descending-ordered belief vector $\omega=(\omega_1, \cdots, \omega_i, \cdots, \omega_k, \cdots, \omega_j, \cdots, \omega_N)$, the greedy policy $\pi^g$ selects the first $k$ elements/channels.  Now consider the following {\em simple} deviation policy $\pi^d$ that selects channels $1, \cdots, i-1, j, i+1, \cdots, k$, where $j>k$. In other words, $\pi^d$ differs from $\pi^g$ in exactly one element: $\omega_j$ instead of $\omega_i$. The one-step deviation produced by this policy is given by 
\begin{eqnarray}
V^{d, \infty}(\omega) = W^\infty(\omega_1, \cdots, \omega_{i-1}, \omega_j, \omega_{i+1}, \cdots, \omega_k, \omega_i, \omega_{k+1}, \cdots, \omega_N)
\end{eqnarray}
since by the definition of $W^\infty(\cdot)$, the RHS operates in exactly the same way as the LHS: it selects the set of channels $1, \cdots, i-1, j, i+1, \cdots, k$, followed by selecting greedily thereafter (note the set of unselected elements are now descending-ordered).  

Our first step is to show that under the stated sufficient condition, we have 
\begin{eqnarray}\label{eqn:infinite_positive_one_step} 
&& W^\infty(\omega_1, \cdots, \omega_{i-1}, \omega_j, \omega_{i+1}, \cdots, \omega_k, \omega_i, \omega_{k+1}, \cdots, \omega_N) \nonumber \\
&\leq& W^\infty(\omega_1, \cdots, \omega_{i-1}, \omega_i, \omega_{i+1}, \cdots, \omega_k, \omega_{k+1}, \cdots, \omega_{j-1}, \omega_j, \omega_{j+1}, \cdots, \omega_N)~, 
\end{eqnarray}
i.e., $\pi^d$ is not a profitable one-step deviation.  We then use this result to show that deviations involving multiple different selections are also not profitable under the same condition, thus proving the lemma. 

To show (\ref{eqn:infinite_positive_one_step}), it suffices to show each of the following chain of (in)equalities under the stated condition: 
\begin{eqnarray}
&& W^\infty(\omega_1, \cdots, \omega_{i-1}, \omega_j, \omega_{i+1}, \cdots, \omega_k, \omega_i, \omega_{k+1}, \cdots, \omega_N) \nonumber \\
&\displaystyle_{=}^{(1)}& W^\infty(\omega_1, \cdots, \omega_{i-1}, \omega_{i+1}, \cdots, \omega_k, \omega_j, \omega_i, \omega_{k+1}, \cdots, \omega_N)  \label{c1}\\
& \displaystyle_{\leq}^{(2)}& W^\infty(\omega_1, \cdots, \omega_{i-1}, \omega_{i+1}, \cdots, \omega_k, \omega_i, \omega_j, \omega_{k+1}, \cdots, \omega_N)  \label{c2}\\
&\displaystyle_{=}^{(1)}& W^\infty(\omega_1, \cdots, \omega_{i-1}, \omega_i, \omega_{i+1}, \cdots, \omega_k, \omega_j, \omega_{k+1}, \cdots, \omega_N)   \label{c3}\\
&\displaystyle_{\leq}^{(3)}& W^\infty(\omega_1, \cdots, \omega_{i-1}, \omega_i, \omega_{i+1}, \cdots, \omega_k, \omega_{k+1}, \omega_j, \cdots, \omega_N)   \label{c4} \\
&\displaystyle_{\leq}^{(3)}& ... \nonumber \\
&\displaystyle_{\leq}^{(3)}& W^\infty(\omega_1, \cdots, \omega_{i-1}, \omega_i, \omega_{i+1}, \cdots, \omega_k, \omega_{k+1}, \cdots, \omega_{j-1}, \omega_j, \omega_{j+1}, \cdots, \omega_N) \label{c5} 
\end{eqnarray} 
%
Note that in each step above the comparison is between switching a neighboring pair of elements.  More specifically, there are three cases: {\bf Case 1} (equalities labeled (1)) involves switching a pair both among the first $k$ elements in the ordered belief vector; {\bf Case 2} (inequality labeled (2)) involves switching a pair at the $k$th and $(k+1)$th positions; {\bf Case 3}  (inequalities labeled (3)) involves switching a pair both outside the first $k$ positions.  These three cases are shown separately below. 

\textbf{Case 1.} 
When both are within the first $k$ elements, there is no difference in either the immediate rewards (both are selected) or the future rewards, so the equality holds trivially. 

\textbf{Case 2.} 
For a given belief vector \rev{(not necessarily descending ordered)} $\omega=(\omega_1, \cdots, \omega_{k}, \omega_{k+1}, \cdots, \omega_N)$ where $\omega_k\geq \omega_{k+1}$, we now compare the difference when switching the order between $\omega_k$ and $\omega_{k+1}$. 
\begin{eqnarray}
&& W^{\infty}(\omega_1, \cdots, \omega_{k}, \omega_{k+1}, \cdots, \omega_N) = \mathbb E  [R_{\pi^g}(\omega_{1},...,\omega_{k-1},\omega_k,\omega_{k+1},...,\omega_{N})] \nonumber \\
&+& \omega_k \cdot \omega_{k+1} \sum_{\bar{l}_{-k} \in \{0,1\}^{k-1}} q(\bar{l}_{-k}; \omega_{-k})\cdot W^{\infty}(p_{11}[\sum_{j=1}^{k-1} l_j], p_{11},p_{11},...,p_{01}[k-1-\sum_{j=1}^{k-1} l_j]) \nonumber \\
 &+&\omega_k\cdot(1- \omega_{k+1}) \sum_{\bar{l}_{-k} \in \{0,1\}^{k-1}} q(\bar{l}_{-k}; \omega_{-k})\cdot W^{\infty}(p_{11}[\sum_{j=1}^{k-1} l_j], p_{11},p_{01},...,p_{01}[k-1-\sum_{j=1}^{k-1} l_j]) \nonumber \\
&+&(1-\omega_k)\cdot \omega_{k+1} \sum_{\bar{l}_{-k} \in \{0,1\}^{k-1}} q(\bar{l}_{-k}; \omega_{-k})\cdot W^{\infty}(p_{11}[\sum_{j=1}^{k-1} l_j], p_{11},...,p_{01}[k-1-\sum_{j=1}^{k-1} l_j],p_{01}) \nonumber \\
&+&(1-\omega_k)\cdot (1-\omega_{k+1}) \sum_{\bar{l}_{-k} \in \{0,1\}^{k-1}} q(\bar{l}_{-k}; \omega_{-k})\cdot W^{\infty}(p_{11}[\sum_{j=1}^{k-1} l_j]], p_{01},...,p_{01}[k-1-\sum_{j=1}^{k-1} l_j]],p_{01}) \nonumber 
\end{eqnarray}
and by switching we have
\begin{eqnarray}
&& W^{\infty}(\omega_1, \cdots, \omega_{k+1}, \omega_{k}, \cdots, \omega_N) = \mathbb E [R_{\pi^g}(\omega_{1},...,\omega_{k-1},\omega_{k+1},\omega_k,...,\omega_{N})] \nonumber \\
&+& \omega_k \cdot \omega_{k+1} \sum_{\bar{l}_{-k} \in \{0,1\}^{k-1}} q(\bar{l}_{-k}; \omega_{-k})\cdot W^{\infty}(p_{11}[\sum_{j=1}^{k-1} l_j], p_{11},p_{11},...,p_{01}[k-1-\sum_{j=1}^{k-1} l_j]) \nonumber \\
 &+&\omega_k \cdot(1- \omega_{k+1}) \sum_{\bar{l}_{-k} \in \{0,1\}^{k-1}} q(\bar{l}_{-k}; \omega_{-k})\cdot W^{\infty}(p_{11}[\sum_{j=1}^{k-1} l_j], p_{11},...,p_{01}[k-1-\sum_{j=1}^{k-1} l_j],p_{01}) \nonumber \\
&+&(1-\omega_k)\cdot \omega_{k+1} \sum_{\bar{l}_{-k} \in \{0,1\}^{k-1}} q(\bar{l}_{-k}; \omega_{-k})\cdot W^{\infty}(p_{11}[\sum_{j=1}^{k-1} l_j], p_{11},p_{01},...,p_{01}[k-1-\sum_{j=1}^{k-1} l_j]) \nonumber \\
&+&(1-\omega_k)\cdot (1-\omega_{k+1})\sum_{\bar{l}_{-k} \in \{0,1\}^{k-1}} q(\bar{l}_{-k}; \omega_{-k})\cdot W^{\infty}(p_{11}[\sum_{j=1}^{k-1} l_j], p_{01},...,p_{01}[k-1-\sum_{j=1}^{k-1} l_j],p_{01})
\nonumber 
\end{eqnarray}
Taking the difference between the immediate rewards we get 
\begin{eqnarray}
\mathbb E [R_{\pi^g}(\omega_{1},...,\omega_k,\omega_{k+1},...,\omega_{N})] - \mathbb E [R_{\pi^g}(\omega_{1},...,\omega_{k+1},\omega_{k},...,\omega_{N})] 
\geq (\omega_k - \omega_{k+1})\underline{\cal R}~. 
\end{eqnarray} 
The difference between the future rewards is given by 
\begin{eqnarray}
&&\beta \cdot (\omega_k - \omega_{k+1})\sum_{\bar{l}_{-k} \in \{0,1\}^{k-1}} q(\bar{l}_{-k}; \omega_{-k})\cdot(W^{\infty}(p_{11}[\sum_{j=1}^{k-1} l_j], p_{11},p_{01},\tau(\omega_{k+2}),...,\tau(\omega_N),p_{01}[k-1-\sum_{j=1}^{k-1} l_j])\nonumber \\
&&~~ -W^{\infty}(p_{11}[\sum_{j=1}^{k-1} l_j], p_{11},\tau(\omega_{k+2}),...,\tau(\omega_N),p_{01}[k-1-\sum_{j=1}^{k-1} l_j],p_{01}) ) \nonumber \\
&\geq& \beta \cdot (\omega_k - \omega_{k+1})\sum_{\bar{l}_{-k} \in \{0,1\}^{k-1}} q(\bar{l}_{-k}; \omega_{-k})\cdot(W^{\infty}(p_{11}[\sum_{j=1}^{k-1} l_j], p_{11},p_{01},\tau(\omega_{k+2}),...,\tau(\omega_N),p_{01}[k-1-\sum_{j=1}^{k-1} l_j])\nonumber \\
&&~~ -W^{\infty}(p_{11}[\sum_{j=1}^{k-1} l_j], p_{11},p_{01},\tau(\omega_{k+3}),...,\tau(\omega_N),p_{01}[k-1-\sum_{j=1}^{k-1} l_j],p_{01}) - (\tau(\omega_{k+2})-p_{01})\Delta_{\infty}) \nonumber \\
&\geq& - \beta \cdot (\omega_k - \omega_{k+1}) (\tau(\omega_{k+2})-p_{01})\Delta_{\infty}
\end{eqnarray} 
where the first inequality comes from the upper bound given in (\ref{eqn:infinite_positive_delta}) and the second from repeated use of the lower bound in Lemma (\ref{lem:boundedness}).  
Thus the total difference in rewards by switching is given by 
$(\omega_k-\omega_{k+1})(\underline{\cal R}-\beta(\tau(\omega_{k+2})-p_{01})\Delta_{\infty})$. 
Since $\tau(\omega_{k+1})\leq p_{11}$, we have
\begin{eqnarray}
\underline{\cal R}-\beta(\tau(\omega_{k+2})-p_{01})\Delta_{\infty} 
&\geq& \underline{\cal R}-\beta(p_{11}-p_{01})\frac{\overline{\cal R}}{1-\beta\delta} \nonumber \\
&\geq& \underline{\cal R}-\delta\frac{\overline{\cal R}}{1-\delta} 
\geq 0 
\end{eqnarray}
under the stated condition on $\delta$. 
\com{If I'm not making a mistake, it seems that the actual sufficient condition I need is $\delta/(1-\delta) \leq \underline{\cal R}/\overline{\cal R}$, without the $p_{11}$...} 


\begin{remark}
\rev{Note in the special case of $k=N-1$, the difference in future rewards by switching is zero, therefore the total difference is always positive without any sufficient condition.  This is consistent with previous results in \cite{Zhao&etal:08TWC} on the optimality of myopic sensing for a two channel case.} 
\end{remark}

\textbf{Case 3.} 
When both elements are outside the first $k$, switching $\omega_i$ with $\omega_{i+1}$, $\omega_i \geq \omega_{i+1}$ results in no difference in the immediate rewards.  Their propagated version, ($\tau(\omega_i), \tau(\omega_{i+1})$), or ($\tau(\omega_{i+1}), \tau(\omega_{i})$) under switching, show up in the future rewards.  As the process continues, this pair will gradually move toward the front of the list, and the movement is exactly the same along each sample path with or without switching.  If the pair continues to be outside the first $k$, the immediate rewards 
remains the same.  If the pair both moves into the first $k$, then the comparison of the future rewards fall within Case 1 examined above.  If the pair moves right into the boundary of the first $k$, with $i$ now at the $k$th position and $i+1$ now at the $k+1$th position (or the other way round under the switched case), then the comparison falls under Case 2 examined above.  Thus this switching under Case 3 is again not profitable.  


%

We have therefore shown that there is no profitable single-element simple deviation under the stated sufficient condition. 
%
For a deviation $\pi^d$ with multiple different elements, a sequence of single-element deviation steps can be easily constructed connecting $\pi^g$ to $\pi^d$, with two successive deviations differing in only one element.  The same argument as above can be used to show that no step  can be profitable under the stated condition, thus proving the lemma. 
\subsection{Proof of Lemma \ref{lem:bound2}}

We prove this by induction. At time $T$, when $j \leq k$ we have
\begin{align}
(x-y)\cdot \underline{\mathcal R} &\leq W_T(\omega_1,...,x,...,\omega_N) \nonumber \\
&-W_T(\omega_1,...,y,...,\omega_N) \leq (x-y)\cdot \overline{\mathcal R}~. 
\end{align}
When $j > k$, we have $W_T(\omega_1,...,x,...,\omega_N) -W_T(\omega_1,...,y,...,\omega_N) = 0$.  Also it is easily verified that  
\begin{align}
\underline{\Delta}_T \leq \min\{\underline{\cal R},0\},  \overline{\Delta}_T \geq  \overline{\cal R} ~. 
\end{align}
The induction basis is thus established. 

Now assume the lemma holds for times $t+1,...,T-1$. Consider time $t$ and again the following cases.

\textbf{Case 1.} $j \leq k$

We have
\begin{align}
&W_t(\omega_1,...,x,...,\omega_N) - W_t(\omega_1,...,y,...,\omega_N) \nonumber \\
&= (x - y)(\mathbb E[R_{\pi^g}(1,\omega_{-j})]- \mathbb E[R_{\pi^g}(0,\omega_{-j})])-(x-y)\cdot \beta \{\sum_{\bar{l}_{-j} \in \{0,1\}^{k-1}} q(\bar{l}_{-j}; \omega_{-j})\nonumber \\
& (W_{t+1}(p_{01}[k-1-\sum_{i=1, i\neq j}^{k} l_i],p_{01},\tau(\omega_N),...,\tau(\omega_{k+1}),p_{11}[\sum_{i=1, i\neq j}^{k} l_i]) \nonumber \\
&~~-W_{t+1}(p_{01}[k-1-\sum_{i=1, i\neq j}^{k} l_i],\tau(\omega_N),...,\tau(\omega_{k+1}),p_{11}[\sum_{i=1, i\neq j}^{k} l_i],p_{11}))\}
\end{align}

Clearly for the immediate rewards we have $\underline{\mathcal R} \leq \mathbb E[R_{\pi^g}(1,\omega_{-j})]- \mathbb E[R_{\pi^g}(0,\omega_{-j})] \leq \overline{\mathcal R}$.  Now consider the future rewards. First note  
\begin{eqnarray}
&& W_{t+1}(p_{01}[k-1-\sum_{i=1, i\neq j}^{k} l_i],p_{01},\tau(\omega_N), ..., \tau(\omega_{k+1}), p_{11}[\sum_{i=1, i\neq j}^{k} l_i]) \nonumber \\
&& ~~ -W_{t+1}(p_{01}[k-1-\sum_{i=1, i\neq j}^{k} l_i], \tau(\omega_N), ..., \tau(\omega_{k+1}), p_{11}[\sum_{i=1, i\neq j}^{k} l_i],p_{11}) \nonumber \\
&=& W_{t+1}(p_{01}[k-1-\sum_{i=1, i\neq j}^{k} l_i],p_{01},\tau(\omega_{N}),..., \tau(\omega_{k+1}), p_{11}[\sum_{i=1, i\neq j}^{k} l_i])\nonumber \\
&& ~~ -W_{t+1}(p_{01}[k-1-\sum_{i=1, i\neq j}^{k} l_i],\tau(\omega_{N}),\tau(\omega_{N}),..., \tau(\omega_{k+1}), p_{11}[\sum_{i=1, i\neq j}^{k} l_i]) \nonumber \\
&& ~~ + W_{t+1}(p_{01}[k-1-\sum_{i=1, i\neq j}^{k} l_i],\tau(\omega_{N}),\tau(\omega_{N}),..., \tau(\omega_{k+1}), p_{11}[\sum_{i=1, i\neq j}^{k} l_i]) \nonumber \\
&& ~~ - W_{t+1}(p_{01}[k-1-\sum_{i=1, i\neq j}^{k} l_i],\tau(\omega_{N}),\tau(\omega_{N-1}),\tau(\omega_{N-1}), ..., \tau(\omega_{k+1}), p_{11}[\sum_{i=1, i\neq j}^{k} l_i]) \nonumber \\
&& ~~ \cdots \nonumber\\
&& ~~ + W_{t+1}(p_{01}[k-1-\sum_{i=1, i\neq j}^{k} l_i],\tau(\omega_{N}),...,\tau(\omega_{k+1}), \tau(\omega_{k+1}), p_{11}[\sum_{i=1, i\neq j}^{k} l_i])\nonumber \\
&& ~~ -W_{t+1}(p_{01}[k-1-\sum_{i=1, i\neq j}^{k} l_i],\tau(\omega_{N}),..., \tau(\omega_{k+1}), p_{11}[\sum_{i=1, i\neq j}^{k} l_i],p_{11}) 
\end{eqnarray}
Applying the induction hypothesis to each pair of the $W_{t+1}$ terms above results in 
\begin{eqnarray}
&& (x-y) \cdot\beta\cdot  (p_{01}-\tau(\omega_N) + \tau(\omega_N) - \tau(\omega_{N-1}) + \cdots - p_{11})\cdot \underline{\Delta}_{t+1} \nonumber \\
&\leq& 
(x - y)\beta \left\{W_{t+1}(p_{01}[k-1-\sum_{i=1, i\neq j}^{k} l_i],p_{01},\tau(\omega_N),..., \tau(\omega_{k+1}), p_{11}[\sum_{i=1, i\neq j}^{k} l_i]) \right. \nonumber \\
&& ~~ \left. -W_{t+1}(p_{01}[k-1-\sum_{i=1, i\neq j}^{k} l_i],\tau(\omega_N), ..., \tau(\omega_{k+1}), p_{11}[\sum_{i=1, i\neq j}^{k} l_i],p_{11}) \right\}\nonumber \\
&\leq&  (x-y) \cdot \beta\cdot(p_{01}-\tau(\omega_N) + \tau(\omega_N) - \tau(\omega_{N-1}) + \cdots - p_{11})\cdot \overline{\Delta}_{t+1}
\end{eqnarray}
Therefore 
\begin{eqnarray}
&& (x-y)\cdot \{\underline{\mathcal R}-\beta\delta\cdot \overline{\Delta}_{t+1}\} \nonumber \\
&\leq& W_t(\omega_1,...,x,...,\omega_N) - W_t(\omega_1,...,y,...,\omega_N) \nonumber \\
&\leq& (x-y)\cdot \{\overline{\mathcal R}-\beta\delta\cdot \underline{\Delta}_{t+1}\}
\end{eqnarray}
If $\eta < 0$ we have
\begin{align}
& \underline{\mathcal R}-\beta\delta\cdot \overline{\Delta}_{t+1}\geq \underline{\mathcal R} - \beta\cdot \delta \overline{\mathcal R} + \beta \delta \frac{1-(\beta\cdot\delta)^{T-t+3}}{1-(\beta\cdot\delta)^{2}}\cdot \eta \geq \frac{1-(\beta\cdot\delta)^{T-t+3}}{1-(\beta\cdot\delta)^{2}}\cdot \eta
= \underline{\Delta}_{t} \\
&\overline{\mathcal R}-\beta\delta\cdot \underline{\Delta}_{t+1} \leq \overline{\mathcal R}-\beta\delta \cdot\eta \leq \overline{\mathcal R} - \frac{1-(\beta\cdot\delta)^{T-t+3}}{1-(\beta\cdot\delta)^{2}}\cdot \eta = \overline{\Delta}_{t} 
\end{align}
If $\eta \geq 0$ we have
\rev{
\begin{align}
\underline{\mathcal R}-\beta\delta\cdot \overline{\Delta}_{t+1} &= \underline{\mathcal R}-\beta\delta\cdot \overline{\mathcal R} = \eta \geq 0 \\
\overline{\mathcal R}-\beta\delta\cdot \underline{\Delta}_{t+1} &=  \overline{\mathcal R} 
\end{align}
} 
In either case the induction step is completed. 

\textbf{Case 2.} $j > k$. 
We have
\begin{align}
& W_t(\omega_1,...,x,...,\omega_N) - W_t(\omega_1,...,y,...,\omega_N)
= \beta \cdot \sum_{\bar{l} \in \{0,1\}^{k}} q(\bar{l}; \bar{\omega}) \nonumber \\
&\cdot(W_{t+1}(p_{01}[k-\sum_{i=1}^{k} l_i],...,\tau(x),...,p_{11}[\sum_{i=1}^{k} l_i]) -
W_{t+1}(p_{01}[k-\sum_{i=1}^{k} l_i],...,\tau(y),...,p_{11}[\sum_{i=1}^{k} l_i]))~. 
\end{align}
Thus 
\begin{align}
\beta \cdot (x-y)\cdot \delta \cdot \underline{\Delta}_{t+1} \leq W_t(\omega_1,...,x,...,\omega_N) 
- W_t(\omega_1,...,y,...,\omega_N) \leq \beta \cdot (x-y)\cdot \delta \cdot \overline{\Delta}_{t+1}
\end{align}
\rev{It can be easily verified that $\underline{\Delta}_{t}\leq \beta\delta\underline{\Delta}_{t+1}$ and 
$\beta\delta\overline{\Delta}_{t+1}\leq \overline{\Delta}_t$, completing the induction step. }

\end{appendix}

\end{document}